\documentclass[12pt,a4paper,amsmath,amssymb]{amsart}

\usepackage[latin1]{inputenc}
\usepackage{amsthm}
\usepackage{latexsym}
\usepackage{amsfonts}
\usepackage{bbm,bm,dsfont}
\usepackage{color} 

\newtheorem{theorem}{Theorem}
\newtheorem{lemma}{Lemma}
\newtheorem{corollary}{Corollary}
\newtheorem{proposition}{Proposition}

\theoremstyle{definition}
\newtheorem{remark}{Remark}
\newtheorem{example}{Example}
\newtheorem{definition}{Definition}

\newcommand{\cal}{\mathcal}

\newcommand{\N}{\mathbb N}
\newcommand{\R}{\mathbb R}

\newcommand{\C}{\mathbb C}

\newcommand{\T}{\mathbb T}

\newcommand{\hi}{\mathcal{H}} 
\newcommand{\hd}{{\mathcal{H}_\oplus}} 
\newcommand{\ki}{\mathcal{K}} 
\newcommand{\li}{\mathcal{L}} 

\newcommand{\lh}{\mathcal{L(H)}} 
\newcommand{\lk}{\mathcal{L(K)}} 
\renewcommand{\th}{\mathcal{T(H)}} 
\newcommand{\tk}{\mathcal{T(K)}} 
\newcommand{\sh}{\mathcal{S(H)}} 
\newcommand{\ph}{\mathcal{P(H)}} 

\newcommand{\tr}[1]{\mathrm{tr}\left[#1\right]} 
\newcommand{\ket}[1]{|#1\rangle} 
\newcommand{\kb}[2]{|#1\,\rangle\langle\,#2|} 

\newcommand{\Mo}{\mathsf{M}} 
\newcommand{\No}{\mathsf{N}} 
\newcommand{\Eo}{\mathsf{E}} 
\newcommand{\Fo}{\mathsf{F}} 
\newcommand{\sfe}{\mathsf{E}} 
\newcommand{\Po}{\mathsf{P}} 
\newcommand{\Jo}{\mathsf{J}} 
\newcommand{\Qo}{\mathsf{Q}}
\newcommand{\Ro}{\mathsf{R}}
\newcommand{\M}{\mathcal M} 
\renewcommand{\O}{\mathrm{Obs}} 
\newcommand{\A}{\mathsf{A}}

\newcommand{\Si}\Sigma
\newcommand{\si}\sigma

\newcommand{\mo}{\mathsf{m}}

\def\<{\langle}
\def\>{\rangle}
\def\d{{\mathrm d}}
\newcommand{\bo}[1]{\mathcal{B}(#1)} 
\newcommand{\fii}{\varphi}

\newcommand{\CHI}[1]{\ensuremath{ \chi\raisebox{-1ex}{$\scriptstyle #1$} }}
\newcommand{\CHII}[1]{\ensuremath{{\hat\chi}\raisebox{-1ex}{$\scriptstyle #1$} }}
\newcommand{\ov}{\overline}

\newcommand{\lin}{{\rm lin}}
\newcommand{\e}{{\bf h}}

\newcommand\Om\Omega

\begin{document}

\title{Optimal quantum observables}

\author{Erkka Theodor Haapasalo}
\address{Department of Nuclear Engineering, Kyoto University, 6158540 Kyoto, Japan}

\author{Juha-Pekka Pellonp\"a\"a}
\address{Turku Centre for Quantum Physics, Department of Physics and Astronomy, University of Turku, FI-20014 Turku, Finland}

\begin{abstract}
Various forms of optimality for quantum observables described as normalized positive operator valued measures (POVMs) are studied in this paper. We give characterizations for observables that determine the values of the measured quantity with probabilistic certainty  or a state of the system before or after the measurement. We investigate observables which are free from noise caused by classical post-processing, mixing, or pre-processing of quantum nature. Especially, a complete characterization of pre-processing and post-processing clean observables is given, and necessary and sufficient conditions are imposed on informationally complete POVMs within the set of pure states. We also discuss joint and sequential measurements of optimal quantum observables.
\newline

\end{abstract}

\maketitle

\section{Introduction}\label{sec:intro}

A normalized positive operator valued measure (POVM) describes the statistics of the outcomes of a quantum measurement and thus we call them as observables of a quantum system. However, some observables can be considered better than the others according to different criteria: The observable may be powerful enough to differentiate between any given initial states of the system or it may be decisive enough to completely determine the state after the measurement no matter how we measure this observable. The observable may also be free from different types of noise either of classical or quantum nature or its measurement cannot be reduced to a measurement of a more informational observable from the measurement of which it can be obtained by modifying either the initial state or the outcome statistics.

We study these various notions of optimality for quantum observables and investigate how they are interrelated. An extensive review of optimal observables and new results especially dealing with post- and pre-processing are given. In this introduction, we approach these problems within a simple setting only considering discrete observables of finite-dimensional quantum systems, formally define the optimality properties outlined above, and characterize observables associated with these properties. In the rest of this paper, we give definitions of optimality in the general case involving also `continuous' observables of infinite-dimensional systems and characterize the optimal observables. However, one can obtain valuable insight in this general case first by looking at the mathematically simpler discussion as follows.

As advertised, let us first consider a POVM $\Mo$ with finitely many values (or outcomes) $\Om=\{x_1,\,x_2\,\ldots,\,x_N\}$ on a finite-dimensional quantum system with the associated Hilbert space $\hi$ (denote $d=\dim\hi<\infty$). This means that we do not need to go into measure theoretical or functional analytical details in this introduction. The POVM $\Mo$ can be viewed as a collection $(\Mo_1,\,\Mo_2,\ldots,\,\Mo_N)$ of  positive semidefinite $d\times d$--matrices $\Mo_i$ such that $\sum_{i=1}^N\Mo_i$ is the identity matrix when (by fixing an orthonormal basis) we identify $\hi$ with $\C^d$ and the bounded operators on $\hi$ with elements of the matrix algebra $\M_d(\C)$. A state of the system is represented as a density matrix $\rho$, that is, a positive semidefinite matrix of trace 1, and the number $p_i=\tr{\rho\Mo_i}\in[0,1]$ is interpreted as the probability of getting an outcome $x_i$ when a measurement of $\Mo$ is performed and the system is in the (initial or input) state $\rho$. Actually, $\Mo$ is a map which assigns to each subset $X$ of $\Om$ a positive matrix $\Mo(X)=\sum_{x_i\in X}\Mo_i$ so that $\tr{\rho\Mo(X)}$ is the probability of getting an outcome belonging to the set $X$. Especially, $\Mo\big(\{x_i\}\big)=\Mo_i$.

We fix a POVM $\Mo$ as above and study its different optimality criteria (in the categories of discrete POVMs in finite dimensions). For that we will need another discrete POVM $\Mo'$, or $(\Mo'_1,\,\Mo'_2,\ldots,\,\Mo'_{N'})$, which acts in a $d'$-dimensional Hilbert space $\hi'\cong\C^{d'}$. Without restricting generality, we will assume that the matrices $\Mo_i$ and $\Mo'_j$ are nonzero.\footnote{If, for instance, $\Mo_i=0$ then $p_i=0$ regardless of the state $\rho$ so the outcome $x_i$ is never obtained and we may replace $\Omega$ by $\Omega\setminus\{x_i\}$ and similarly remove all outcomes related to zero matrices.}

Write $\Mo_i=\sum_{k=1}^{m_i}\lambda_{ik}\kb{\fii_{ik}}{\fii_{ik}}=\sum_{k=1}^{m_i}\kb{d_{ik}}{d_{ik}}$ where the eigenvectors $\fii_{ik}$, $k=1,\ldots,\,m_i$, form an orthonormal set, the eigenvalues $\lambda_{ik}$  are positive (and bounded by 1), and $d_{ik}=\sqrt{\lambda_{ik}}\fii_{ik}$. We say that $m_i$ is the multiplicity of the outcome $x_i$ or the rank of $\Mo_i$, and $\Mo$ is of rank 1 if $m_i=1$ for all $i=1,\ldots,N$. Recall that rank-1 observables have many important properties \cite{BuKeDPeWe2005,HeWo,Part2,Pell,Pell',Pell2}. For example, their measurements break entanglement completely between the system and its environment \cite{Pell}. One can define a (maximal) rank-1 refinement POVM $\Mo^{\bm1}$ of $\Mo$ via $\Mo^{\bm1}_{ik}=\kb{d_{ik}}{d_{ik}}$, $i=1,\ldots,N$, $k=1,\ldots,m_i$. Now $\Mo^{\bm1}$ has $N^1=\sum_{i=1}^N m_i$ outcomes and $p_i=\tr{\rho\Mo_i}=\sum_{k=1}^{m_i}\tr{\rho\Mo^{\bm1}_{ik}}$ (i.e.\ $\Mo$ is a relabeling of $\Mo^{\bm1}$) thus showing that any measurement of $\Mo^{\bm1}$ can be viewed as a measurement of $\Mo$, the so-called complete measurement, since the value space of $\Mo^{\bm1}$ `contains' also the multiplicities $k\le m_i$ of the measurement outcomes $x_i$ of $\Mo$, see \cite{Pell,Pell',Pell2} for further properties of complete measurements.

Let then $\hd$ be a Hilbert space spanned by an orthonormal basis $e_{ik}$ where $i=1,\ldots,N$ and $k=1,\ldots,m_i$. Obviously, $\dim\hd=N^1$.
Define a discrete normalized projection valued measure (PVM) $\Po=(\Po_1,\ldots,\Po_N)$ of $\hd$ via
$
\Po_i=\sum_{k=1}^{m_i}\kb{e_{ik}}{e_{ik}}
$
so that  $\Po_i\hd$ is spanned by the vectors $e_{ik}$, $k=1,\ldots,m_i$, and we may write (the direct sum)
$
\hd=\bigoplus_{i=1}^N(\Po_i\hd).
$
Define an isometry $J:\,\hi\to\hd$, 
$
J=\sum_{i=1}^N\sum_{k=1}^{m_i}\kb{e_{ik}}{d_{ik}}
$
for which
$
J^*\Po_iJ=\Mo_i.
$
Hence, $\big(\hd,J,\Po\big)$ is a Na\u{\i}mark dilation of $\Mo$.\footnote{The dilation is minimal, that is, the span of vectors $\Po_i J\phi$, $i=1,\ldots,N$, $\phi\in\hi$, is the whole $\hd$. Indeed, this follows immediately from equation $\psi=\sum_{i=1}^N\sum_{k=1}^{m_i}\<e_{ik}|\psi\>e_{ik}=
\sum_{i=1}^N\sum_{k=1}^{m_i}\<e_{ik}|\psi\>\lambda_{ik}^{-1}\Po_i J d_{ik}$ where $\psi\in\hd$.} Note that one can identify $\hi$ with a (closed) subspace $J\hi$ of $\hd$, equipped with the projection $JJ^*$ from $\hd$ onto $J\hi$, and we may briefly write $\hd=\hi\oplus\hi^\perp$. Especially, any state $\rho$ of $\hi$ can be viewed as a state $J\rho J^*$ of the bigger space $\hd$. By using this interpretation, any measurement of $\Po$ in the subsystem's state can be viewed as a measurement of $\Mo$ via $p_i=\tr{\rho\Mo_i}=\tr{J\rho J^*\Po_i}$. Finally, we note that $\Mo$ is a PVM if and only if $J$ is unitary (i.e.\ $\{d_{ik}\}_{i,k}$ is an orthonormal basis of $\hi$). In this case one can identify $\hd$ with $\hi$ and $\Po$ with $\Mo$ e.g.\ by setting $e_{ik}=d_{ik}$.

\begin{remark}\rm\label{remu1}
Let $\big(\hd,J,\Po\big)$ be a (minimal) Na\u{\i}mark dilation of the POVM $\Mo$ as above.
Without restricting generality, one can pick any orthonormal basis $\{e_n\}_{n=1}^\infty$ of an infinite-dimensional  Hilbert space $\hi_\infty$ and choose $e_{ik}=e_i\otimes e_k\in\hi_\infty\otimes\hi_\infty$ so that $\hd$ becomes a (closed) subspace of $\hi_\infty\otimes\hi_\infty$
and $\Po_i=\kb{e_{i}}{e_{i}}\otimes\sum_{k=1}^{m_i}\kb{e_{k}}{e_{k}}\le\Po'_i\otimes I_M$ where
$\Po'_i=\kb{e_{i}}{e_{i}}$, $i=1,\ldots,N$, constitutes a rank-1 PVM $\Po'$ in an $N$-dimensional space $\hi_N=\lin\{e_i\,|\,i=1,\ldots,N\}$ and $I_M=\sum_{k=1}^{M}\kb{e_{k}}{e_{k}}$ is the identity operator of $\hi_M=\lin\{e_i\,|\,i=1,\ldots,M\}$ where $M=\max_{i\le N}\{m_i\}$. 
In addition, $J$ can be interpreted as an isometry from $\hi$ into $\hi_N\otimes\hi_M$ by the same formula $J=\sum_{i=1}^N\sum_{k=1}^{m_i}\kb{e_i\otimes e_k}{d_{ik}}$ and we have\footnote{Hence, $\big(\hi_N\otimes\hi_M,J,\Po'\otimes I_M\big)$ is a 
Na\u{\i}mark dilation of $\Mo$, which is minimal if and only if $m_i=M$ for all $i=1,\ldots,N$.} 
$\Mo_i=J^*(\Po'_i\otimes I_M)J=\Phi_J(\Po_i')$ where 
$\Phi_J$ is a (completely positive) Heisenberg channel, $\Phi_J(B)=J^*(B\otimes I_M)J=\sum_{s=1}^M\A_s^*B\A_s$ where $B$ is an $N\times N$--matrix (i.e.\ $B\in\M_N(\C)$).\footnote{Note that the Kraus operators $\A_s=\sum_{i=1}^N\kb{e_i}{d_{is}}$ are linearly independent, i.e.\ the Kraus decomposition of $\Phi_J$ is minimal. In addition,  the corresponding Schr\"odinger channel $(\Phi_J)_*$ transforms a $d\times d$--state $\rho$ to the $N\times N$--state $\rho'=\sum_{s=1}^M\A_s\rho\A_s^*$ and $p_i=\tr{\rho\Mo_i}=\tr{\rho'\Po_i'}=\<e_i|\rho'|e_i\>$ holds.}
\end{remark}

Davies and Lewis \cite{DaLe} introduced the concept of instrument which turned out to be crucial in developing quantum measurement theory since,
besides measurement statistics, it also describes the conditional state changes due to a quantum measuring process. For example, if the measurement outcome set is finite, $\Omega=\{x_1,\ldots,x_N\}$,
then any (Schr\"odinger) instrument ${\mathcal I}$ describing a measurement of $\Mo$ (with the outcomes $\Omega$), can be viewed as a collection of (completely positive) operations\footnote{For any $i$, ${\mathcal I}_i(\rho)=\sum_s \A_{is}\rho\A_{is}^*$ (a Kraus decomposition) and the dual (Heisenberg) operation is $\mathcal J_i(B)={\mathcal I}^*_i(B)=\sum_s \A_{is}^*B\A_{is}$ where $B$ is any $d\times d$--matrix.}  ${\mathcal I}_i$ on $\M_d(\C)$ such that $\sum_{i=1}^N\tr{{\mathcal I}_i(\rho)}=1$ for any state $\rho$. Now $\mathcal I$ transforms an input state $\rho$ to a (nonnormalized) output state ${\mathcal I}_i(\rho)$ if $x_i$ is obtained. In addition, $\mathcal I$ defines the measurement outcome probabilities $p_i$ and the corresponding POVM $\Mo$ via $p_i=\tr{{\mathcal I}_i(\rho)}=\tr{\rho\Mo_i}$.\footnote{In other words, using the Kraus decompositions of the operations, $\Mo_i=\sum_s \A_{is}^*\A_{is}$.} Note that $\rho\mapsto \sum_{i=1}^N{\mathcal I}_i(\rho)$ is a (Schr\"odinger) channel which transforms any state of the system to another state of the same system. More generally, a quantum channel is a completely positive trace-preserving (cptp) linear map between state spaces associated to quantum systems (with possibly different Hilbert spaces $\hi$, $\hi'$) so that channels transmit quantum information between different systems. Similarly, with possibly different input and output spaces $\hi\cong\C^d$ and $\hi'\cong\C^{d'}$, one may also assume an initial state of $\hi$ to transform into conditional states of $\hi'$ as a result of the measurement prompting to describe the measurement through an instrument $\mathcal I$ with Schr\"odinger operations $\mathcal I_i:\M_d(\C)\to\M_{d'}(\C)$.

Now we are ready to introduce the following six optimality criteria for $\Mo$:
\begin{enumerate}
\item[(1a)] {\it $\Mo$ determines the future of the system (completely)} if each instrument ${\mathcal I}$ implementing $\Mo$ is nuclear (or preparatory), i.e.,\ of the form ${\mathcal I}_i(\rho)=p_i\sigma_i$ where $\sigma_i$'s are density matrices (of any fixed output Hilbert space $\hi'$) which do no depend on the input state $\rho$. If the outcome $x_i$ is obtained with the nonzero probability $p_i=\tr{\rho\Mo_i}$ then the output system is in the $\rho$-independent state $\sigma_i$ after the measurement. It can be shown that {\it $\Mo$ determines the future if and only if $\Mo$ is of rank 1,} i.e.\ each $\Mo_i$ is of the form $\kb{d_i}{d_i}$ where $d_i\in\hi$ \cite{HeWo,Part2}.
\item[(1b)] {\it $\Mo$ is post-processing maximal (post-processing clean)} if the condition $\Mo_i=\sum_{j=1}^{N'}p'_{ji}\Mo'_j$ for all $i$ (where $(p'_{ji})$ is $N'\times N$--probability matrix and $\Mo'=(\Mo'_1,\ldots,\Mo'_{N'})$ is a POVM of the same Hilbert space $\hi'=\hi$) implies that $\Mo_j'=\sum_{i=1}^N p_{ij}\Mo_i$ for all $j$ where
$(p_{ij})$ is $N\times N'$--probability matrix.\footnote{Recall that $(p_{ij})$ is a probability (or stochastic or Markov) matrix  if $p_{ij}\ge 0$ and $\sum_j p_{ij}=1$ for all $i$. The numbers $p_{ij}$ are transition probabilities and $\Mo'$ is said to be a smearing of $\Mo$ if $\Mo_j'=\sum_i p_{ij}\Mo_i$ holds. In this case, $\Mo$ and $\Mo'$ are jointly measurable, a joint observable being $\No_{ij}=p_{ij}\Mo_i$.} The condition $\Mo_i=\sum_j p'_{ji}\Mo'_j$ yields $p_i=\tr{\rho\Mo_i}=\sum_j p'_{ji}\tr{\rho\Mo'_j}$ thus showing that, instead of measuring $\Mo$, one can measure $\Mo'$ in the same state $\rho$ and then classically post-process the data by using the matrix $(p'_{ij})$. Post-processing clean POVMs are free from this type of classical noise and it is easy to show that {\it $\Mo$ is post-processing clean if and only if $\Mo$ is of rank 1} \cite[Theorem 3.4]{DoGr97}.
\item[(2a)] {\it $\Mo$ determines the past of the system} if it is {\it informationally complete,} i.e.\ the measurement outcome statistics $(p_i)_{i=1}^N$ determines the input state $\rho$, i.e.\ the condition $\tr{\rho\Mo_i}=\tr{\rho'\Mo_i}$ for all $i$ implies that $\rho'=\rho$. Clearly, {\it $\Mo$ determines the past of the system if and only if $N\ge d^2$ and any $d\times d$--matrix $B$ can be written as a linear combination of matrices $\Mo_i$, $i=1,\ldots,N$} \cite[Prop.\ 18.1]{kirja}. We will construct later an informationally complete (extreme) rank-1 POVM $\Mo$ with the minimum number $N=d^2$ of outcomes. Generally, an informationally complete POVM need not be rank-1 but {\it if $\Mo$ is informationally complete then its rank-1 refinement $\Mo^{\bm1}$ is also informationally complete} \cite{Pell}.
\item[(2b)] {\it $\Mo$ is extreme} if it is an extremal point of the convex set of all (discrete) POVMs of $\hi$, i.e.\ $\Mo=\frac12\Mo'+\frac12\Mo''$ implies $\Mo'=\Mo$. Thus, extreme observables describe statistics of the pure quantum measurements, free from any classical randomness
due to fluctuations in the measuring procedure (in the same way as pure states describe preparation procedures without classical randomness). One can show that {\it $\Mo$ is extreme if and only if the matrices $\kb{d_{ik}}{d_{i\ell}}$, $i=1,\ldots,N$, $k,\ell=1,\ldots,m_i$, are linearly independent} \cite[Theorem 2.4]{Partha}, \cite[Theorem 2]{DALoPe}. Especially, if $\Mo$ is rank-1, i.e.\ $\Mo_i=\kb{d_i}{d_i}$, $d_i\ne0$, then it is extreme if and only if the matrices $\Mo_i$ are linearly independent. In this case, it is informationally complete if and only if $N=d^2$. Trivially, {\it if $\Mo$ is extreme then its rank-1 refinement $\Mo^{\bm1}$ is also extreme} thus\footnote{Since $d^2=N^1=\sum_{i=1}^N m_i\ge N$ and $N\ge d^2$ imply $m_i\equiv1$ and $N=d^2$.} showing that {\it any extreme informationally completely POVM is necessarily of rank 1.} Finally, we note that PVMs are automatically extreme. 
\item[(3a)] {\it $\Mo$ determines its values} $x_i$ if each $\Mo_i$ is of (operator) norm 1, i.e. $\Mo_i$ has the eigenvalue $1$ (with the unit eigenvector $\fii_i$). In this case, for any outcome $x_j$ one can pick a state $\rho=\kb{\fii_j}{\fii_j}$ such that $p_i=\tr{\rho\Mo_i}=\delta_{ij}$ for all $i$, i.e.\ the observable $\Mo$ has the value $x_j$ in the state $\rho$ with probabilistic certainty.\footnote{Note that this holds only in finite dimensions. Generally, a norm-1 (effect) operator can have a fully continuous spectrum (i.e.\ no eigenvalues at all). However, even in such a case, for each $j$ and any $\varepsilon\in(0,1)$, there is a state $\rho=\kb{\fii_j}{\fii_j}$ such that $\tr{\rho\Mo_j}=1-\varepsilon$.} Clearly, a rank-1 norm-1 POVM is a PVM\footnote{This holds for discrete observables. As a counterexample, consider the canonical phase observable which is rank-1 norm-1 `continuous' POVM but not projection valued.} and any PVM is of norm-1.
\item[(3b)] {\it $\Mo$ is pre-processing maximal (pre-processing clean)} if the condition $\Mo_i=\Phi(\Mo'_i)$ for all $i$ (where $\Phi:\M_{d'}(\C)\to\M_d(\C)$ is a Heisenberg channel and $\Mo'$ is a POVM on the possibly different Hilbert space $\hi'\cong\C^{d'}$ with $N'=N$ outcomes) implies that $\Mo'_i=\Theta(\Mo_i)$ for all $i$ where $\Theta:\M_d(\C)\to\M_{d'}(\C)$ is some Heisenberg channel. The condition $\Mo_i=\Phi(\Mo'_i)$ can be written in the form $p_i=\tr{\rho\Mo_i}=\tr{\rho\Phi(\Mo'_i)}=\tr{\Phi_*(\rho)\Mo'_i}$ so that to get the probabilities $p_i$ one can equally well measure $\Mo'$ in the state $\Phi_*(\rho)$, i.e.\ $\Mo'$ is `better' measurement in this sense and $\Mo$ is obtained from it by adding quantum noise in $\rho$ (characterized by the channel $\Phi$). Hence,
pre-processing clean observables are free from this type of quantum noise. Since, using the Na\u{\i}mark dilation $(\hi_\oplus,\Po,J)$ of $\Mo$, $\Mo_i=J^*\Po_iJ=\Phi(\Po_i)$, where $\Phi$ is the (rank-1) isometry channel $J^*(\,\cdot\,)J$, to show that $\Mo$ is pre-processing clean, one must find a channel $\Theta$ such that $\Po_i=\Theta(\Mo_i)$ holds\footnote{Actually, Remark \ref{remu1} shows that $\Mo_i=\Phi_J(\Po_i')$ must be connected to a {\it rank-1} PVM $\Po'$ via some channel.} and thus\footnote{$1=\|\Po_i\|=\|\Theta(\Mo_i)\|\le\|\Theta\|\,\|\Mo_i\|=\|\Mo_i\|\le1$ implies $\|\Mo_i\|=1$.} each $\Mo_i$ is of norm 1, i.e.\ $\Mo$ determines its values. We will show that, in finite dimensions, {\it pre-processing clean POVMs are exactly norm-1 POVMs} and exactly of the form $\Mo_i=\Eo_i\oplus\Fo_i$ where $\Eo$ is a PVM ($\Eo_i\ne 0$ for all $i$) and $\Fo$ a POVM ($\Theta(\Fo_i)=0$ for all $i$) acting on orthogonal subspaces of $\hi$. Hence, for any pre-processing clean POVM $\Mo$, there exists a projection (onto a subspace) such that the projected POVM $\Eo$ is projection valued. Especially, PVMs are pre-processing clean \cite{Pe11}.
\end{enumerate}

We have seen that `optimal observables' must be of rank 1, see (1a) and (1b) above. Moreover, observables satisfying (2a) or (2b) can be maximally refined into rank-1 observables which share the same optimality criteria as the original POVMs. If $\Mo$ is rank-1 then the map $\C^N\ni(c_1,\ldots,c_2)\mapsto\sum_{i=1}^Nc_i\Mo_i$ is surjective iff (2a) holds and injective iff (2b) holds. We will construct a POVM for which all conditions (1a), (1b), (2a), and (2b) hold. PVMs are optimal observables in the sense of (3a) and (3b). In addition, the rank-1 refinement of a PVM is also projection valued (and of norm-1). However, it is easy to construct a norm-1 POVM whose rank-1 refinement is not of norm 1. For example, in $\C^3$ (with the basis $\ket0$, $\ket1$, $\ket2$), one can define 2-valued norm-1 POVM $\Mo_1=\kb11+\frac13\kb00$, $\Mo_2=\kb22+\frac23\kb00$ whose refinement $\Mo^{\bm1}$ has an effect $\Mo^{\bm1}_{12}=\frac13\kb00$ of norm $\frac13$. Note that $\Mo_i=\Eo_i\oplus\Fo_i$ where $\Eo_1=\kb11$, $\Eo_2=\kb22$ constitutes a PVM in a 2-dimensional space, and $\Fo_1=\frac13\kb00$, $\Fo_2=\frac23\kb00$. 

To conclude, there are essentially two sorts of  optimal observables: rank-1 PVMs and extreme informationally complete POVMs. Since they are extreme (2b) and rank-1, i.e.\ post processing clean (1b), they are free from classical noise due to the mixing of measurement schemes or data processing.
Moreover, they determine the future of the system (1a). It is easy to show that a pre-processing clean POVM (e.g.\ a PVM) cannot be informationally complete and vice versa,\footnote{Since a pre-processing clean POVM has at most $N=d$ outcomes and an informationally complete POVM has at least $N=d^2$ outcomes.} i.e.\ an informationally complete POVM is never free from quantum noise. Moreover, {\it the determination of the past (2a) and the values (3a) are complementary properties.} However, when one assumes that only a restricted class of states (related to a subspace $\hi\subseteq\hd$) can be determined completely then these complementary properties can be combined as follows:

One can pick a $d^2$-outcome extreme informationally complete rank-1 POVM $\Mo_i=\kb{d_i}{d_i}$, $i=1,\ldots,d^2$, and its minimal Na\u{\i}mark dilation with the rank-1 PVM $\Po_i=\kb{e_i}{e_i}$
acting in a $d^2$-dimensional space $\hd=\hi\oplus\hi^\perp$ with the orthogonal basis $\{e_i\}_{i=1}^{d^2}$ (recall that $d=\dim\hi$). Now, for any subsystem's state $\rho$, one gets $p_i=\<d_i|\rho|d_i\>=\<e_i|J\rho J^*|e_i\>$ and $\Po_i$ is informationally complete only within the set of states of the subspace $\hi$. Instead of measuring $\Mo$ one can prepare a state of $\hi\cong J\hi$ and then perform a measurement of $\Po$ to get probabilities $p_i$ and posterior states $\sigma_i$ (see item (1a) above) since the nuclear instrument ${\mathcal I}_i(\rho)=\<d_i|\rho|d_i\>\sigma_i$ implementing $\Mo$ can be trivially extended to an instrument $\overline{\mathcal I}$ of $\Po$ via $\overline{\mathcal I}_i(\overline\rho)=\<e_i|\overline\rho|e_i\>\sigma_i$ where $\overline\rho$ is a state of $\hd$. Below we study sequential and joint measurements of optimal POVMs with other observables.

Let $\Mo=(\Mo_i)_{i=1}^N$ and $\Mo'=(\Mo'_j)_{j=1}^{N'}$ be POVMs as in the beginning of this introduction, and let $\cal I=(\cal I_i)$ be an instrument implementing $\Mo$ (i.e.\ any $\mathcal I_i:\M_d(\C)\to\M_{d'}(\C)$ is of the form ${\mathcal I}_i(\rho)=\sum_s \A_{is}\rho\A_{is}^*$ and $\Mo_i=\sum_s \A_{is}^*\A_{is}$). Suppose then that one measures first $\Mo$ in the state $\rho$ (described by $\cal I$) and then $\Mo'$ in the transformed (conditional) state $p_i^{-1}{\cal I}_i(\rho)$ if the outcome $x_i$ is obtained (with the probability $p_i>0$) in the first measurement of $\Mo$. This sequential measurement can be described by a joint POVM $\Jo=(\Jo_{ij})$ where
$\Jo_{ij}={\cal I}^*_i(\Mo'_j)$ since the conditional probability is $\tr{\rho\Jo_{ij}}=\tr{p_i^{-1}{\cal I}_i(\rho)\Mo'_j} p_i$. Hence, {\it a sequential measurement of $\Mo$ and $\Mo'$ can be interpreted as a joint measurement} of $\Mo$ and the disturbed POVM $\Mo''$, $\Mo''_j=\sum_iJ_{ij}=\Phi(\Mo'_j)$ of the same Hilbert space (here $\Phi=\sum_i{\cal I}^*_i$ is the total Heisenberg channel of $\cal I$).

Indeed, any POVMs $\Mo=(\Mo_i)$ and $\Mo''=(\Mo''_j)$ (of the same Hilbert space $\hi''=\hi$) are jointly measurable if there exists a POVM 
$\No=(\No_{ij})$ such that $\Mo$ and $\Mo''$ are the margins of $\No$, i.e.,\ $\Mo_i=\sum_{j=1}^{N''} \No_{ij}$ and $\Mo''_j=\sum_{i=1}^N \No_{ij}.$
If $\big(\hd,J,\Po\big)$ is a (minimal) Na\u{\i}mark dilation of $\Mo$ then it is easy to show\footnote{Since $\No_{ij}\le \sum_{j=1}^{N''} \No_{ij}=J^*\Po_i J$ implies the existence of $\Po_{ij}\le\Po_i$ ($m_i\times m_i$--identity matrix).} that $\No_{ij}=J^*\Po_{ij}J$ where $\Po_{ij}$ is a (unique) positive semidefinite $m_i\times m_i$--matrix such that $\sum_j\Po_{ij}=\Po_i$. Hence, for each $i\le N$, the map $j\mapsto \Po_{ij}$ is a POVM\footnote{This POVM acts in the subspace $\Po_i\hd$ and is normalized there.} with $N''$ values so that there is a channel $\Phi_i$ such that $\Phi_i(\Po'_j)=\Po_{ij}$, see Remark \ref{remu1}. Here $\Po'=(\Po'_j)_{j=1}^{N''}$ is a fixed rank-1 PVM acting in a minimal\footnote{If, say, $\Mo''_k=0$ then $\No_{ik}\le\sum_{i'=1}^N \No_{i'k}= \Mo''_k$ yields $\No_{ik}=0$ and thus $\Po_{ik}=0$ for all $i\le N$.  Hence, $\Po'_k=0$ and $N''$ is the number of nonzero effects $\Mo_j''$ since usually we assume that $\Mo_j''\ne 0$ for all $j=1,\ldots,N''$.} $N''$-dimensional Hilbert space $\hi_{N''}\cong\C^{N''}$. Define an instrument $\cal I$ by $\cal I_i^*(B)=J^*\Phi_i(B)J$, $B\in\M_{N''}(\C)$. Clearly, $\No_{ij}=\cal I_i^*(\Po_j')$ and, thus, any {\it joint measurement of $\Mo$ and $\Mo''$ can be interpreted as a sequential measurement} of $\Mo$ followed by a rank-1 PVM $\Po'$. Note that $\Mo''_j=\Phi'(\Po'_j)$
where $\Phi'(B)=J^*\sum_i\Phi_i(B)J$.

In conlusion, a sequential measurement of $\Mo$ and $\Mo'$ defines a joint obsevable $\Jo$ with the margins $\Mo$ and $\Mo''=\Phi(\Mo')$. If we put $\No=\Jo$ above, we see that this measurement of $\Jo$ can be interpreted as a new  sequential measurement of $\Mo$ and a rank-1 PVM $\Po'$. In addition, $\Mo''=\Phi'(\Po')$. Thus, the latter observable in a sequential set-up can be assumed to be very optimal: free from both classical and quantum noise. Next we study how the optimality criteria (1a)---(3b) affect the joint measurability of an optimal observable with other observables.

\begin{enumerate}
\item
If $\Mo$ is rank-1 ($m_i\equiv1$) then any $\Po_{ij}$ is a $1\times1$--matrix, i.e.\ a number $p_{ij}$, and $\No_{ij}=p_{ij}\Mo_i$ is rank-1 (and a post-processing of $\Mo$). Since $(p_{ij})$ is a probability matrix, also $\Mo''$ is a smearing (post-processing) of $\Mo$, i.e.\ $\Mo_j''=\sum_i p_{ij}\Mo_i$, see item (1b) above. Moreover, the $\Mo$-compatible instrument $\cal I$ is nuclear (1a) and $\No_{ij}=\Jo_{ij}=\tr{\sigma_i\Mo_j'}\Mo_i$ from where one can read the transition probabilities $p_{ij}=\tr{\sigma_i\Mo_j'}$ (where the states $\sigma_i$ determines $\cal I$ completely) showing that, if one gets $x_i$ in the first $\Mo$-measurement, then the instrument `prepares' the post-measurement state $\sigma_i$ which is the input state for the second $\Mo'$-measurement giving the probability distribution $p_{ij}$, $j=1,\ldots,N'$, and $\tr{\rho\Jo_{ij}}=p_ip_{ij}$ where $p_i=\tr{\rho\Mo_i}$. Hence, after a measurement of a rank-1 POVM there is no need to perform any extra measurements to get more information. It should be stressed that, even if $\Mo''_j=\Phi'(\Po'_j)$, the entanglement breaking channel $\Phi'$ (associated to a nuclear instrument\footnote{Since $\Po'_j=\kb{e_j}{e_j}$ one can define states $\sigma'_i=\sum_j p_{ij}\kb{e_j}{e_j}$ and a nuclear instrument ${\mathcal I}'_i(\rho)=\tr{\rho\Mo_i}\sigma'_i$ (or ${{\mathcal I}'_i}^*(B)=\tr{B\sigma'_i}\Mo_i$) such that ${{\mathcal I}'_i}^*(\Po_j')=p_{ij}\Mo_i=\No_{ij}$ and $\Phi'(B)=\sum_i{{\cal I}'_i}^*(B)=\sum_i\tr{B\sigma'_i}\Mo_i$.} ${\mathcal I}'$ of $\Mo$) \cite{Part2} adds so much quantum noise to the rank-1 PVM $\Po'$ that it becomes the fuzzy version $\Mo''$ of $\Mo$. Hence, in this sequential measurement, the latter observable $\Mo''$, which arises as the second marginal of the joint observable $\Jo$, is obtained both through adding classical noise to the observable $\Mo$ first measured (i.e.,\ as a post-processing $\Mo_j''=\sum_ip_{ij}\Mo_i$) and through adding quantum noise to the observable $\Mo'$ actually measured after $\Mo$ in the form of the pre-processing $\Mo''_j=\Phi(\Mo'_j)=\sum_i\tr{\sigma_i\Mo'_j}\Mo_i$. The same results naturally apply in the situation where we modify the measurement of the first observable $\Mo$ and measure some rank-1 PVM $\Po'$ after it to obtain $\Mo''$ as the second marginal. In this case $\Mo''$ is a classical smearing of the rank-1 $\Mo$ and a quantum smearing of the rank-1 PVM $\Po'$.

\item
If $\Mo$ is informationally complete (2a) then $\No=\Jo$ is also informationally complete.\footnote{If one gets probabilites $\tr{\rho\No_{ij}}$ then one can solve $\rho$ from the probabilities $p_i=\tr{\rho\Mo_i}=\sum_{j=1}^{N''} \tr{\rho\No_{ij}}$.}
Hence, trivially, if already $\Mo$ determines the past, then its subsequent measurements cannot increase the (already maximal) state distinguishing power. Suppose now that $\No_{ij}=\Jo_{ij}=\mathcal I_i^*(\Mo'_j)$ for some instrument $\mathcal I$ measuring the informationally complete $\Mo$ and some subsequently measured $\Mo'$ giving rise to the second marginal $\Mo''_j=\Phi(\Mo'_j)$ for the total channel $\Phi=\sum_i\mathcal I_i^*$. If we also assume that $\Mo''$ is informationally complete, i.e.,\ we jointly measure two informationally complete observables in a sequential setting, then the Heisenberg channel $\Phi$ is surjective (in this finite-dimensional case) and the corresponding Schr\"odinger channel $\Phi_*=\sum_i\mathcal I_i$ is injective.\footnote{Since $\Mo''$ is informationally complete, the map $\M_d(\C)\ni\rho\mapsto\big(\tr{\rho\Mo''_j}\big)_j$ is injective, and, since this map is the composition of the maps $\Phi_*$ and $\sigma\mapsto\big(\tr{\sigma\Mo'_j}\big)_j$, as the first of these maps, $\Phi_*$ has to be injective.}

If $\Mo$ is extreme (2b) then $\No$ is the unique joint POVM which has the margins $\Mo$ and $\Mo''$.\footnote{Since, for $\No_{ij}=J^*\Po_{ij}J$ and $\No'_{ij}=J^*\Po'_{ij}J$, the condition $\Mo''_j=\sum_{i=1}^N \No_{ij}=\sum_{i=1}^N \No'_{ij}$ can be written in the form $J^*D_jJ=0$ where $D_j=\oplus_{i=1}^N(\Po_{ij}-\Po'_{ij})$. If $\Mo$ is extreme then $D_j\equiv0$, i.e.\ $\Po_{ij}\equiv\Po'_{ij}$, and $\No=\No'$.} If, in addition, $\Mo$ is rank-1 then $\No_{ij}=p_{ij}\Mo_i$ and $\Mo_j''=\sum_i p_{ij}\Mo_i$ and we have the chain of bijections: $\No\mapsto\Mo''\mapsto(p_{ij})\mapsto\No$.

\item
If $\Mo=\Po$ is a PVM, i.e.\ $\Po_k\Po_\ell\equiv\delta_{k\ell}\Po_k$, (and thus $\hd=\hi$ and $J$ is the identity map) then $\No_{ij}=\Po_{ij}$ where each map $j\mapsto\Po_{ij}$ is a (subnormalized) POVM which commutes with $\Po$, i.e.\ $\Po_{ij}\Po_k\equiv\Po_k\Po_{ij}$, since $\Po_{ij}\le\Po_i$. Thus, any $\Mo''$ compatible with a PVM $\Po$ commutes with $\Po$. If, moreover, $\Po$ is of rank-1 then $\Po_{ij}=p_{ij}\Po_i$. Note that $\No$ (or $\Mo''$) needs not to be a PVM or even of norm 1 (e.g.\ consider $\Po_1=\kb11$, $\Po_2=\kb22=\Po_{22}$, $\Po_{11}=\Po_{12}=\frac12\kb11=\Mo_2''$, $\Mo_1''=\frac12\kb11+\kb22$ in $\C^2$).
\end{enumerate}

In this paper, we generalize the above results to the case of arbitrary observables (with sufficiently `nice' value spaces) acting in separable Hilbert spaces. For example, consider the single-mode optical field with the Hilbert space $\hi$ spanned by the photon number states $|n\>$, $n=0,1,2,\ldots,$ associated with the number operator $N=a^*a=\sum_{n=0}^\infty n\kb n n$ where $a=\sum_{n=0}^\infty\sqrt{n+1}\kb n{n+1}$. Define the position and momentum operators $Q=\frac{1}{\sqrt{2}}{(a^*+a)}$ and $P=\frac{i}{\sqrt{2}}{(a^*-a)}$ which, in the position representation, are the usual multiplication and differentiation operators,
$(Q\psi)(x) = x\psi(x)$ and $(P\psi)(x)= -i\,{\d\psi(x)}/{\d x}$ (we set $\hbar=1$). Define the Weyl operator (or the displacement operator of the complex plane)\footnote{Recall that Weyl operators are associated to a unitary representation of the Heisenberg group $\mathbb H$, or to a projective representation of the additive group $\C\cong\R^2$.}  
$D(z)=e^{{z a^*-\overline z a}}$, $z\in\C$. Let $q,\,p\in\R$ and $z=(q+i p)/\sqrt2$. Then
$$
D(q,p)=D(z)=e^{ipQ-iqP}=e^{-iqp/2}e^{ipQ}e^{-iqP}=e^{iqp/2}e^{-iqP}e^{ipQ},
$$
i.e.,\ for all $\psi\in\hi\cong L^2(\R)$, $(e^{ipQ}\psi)(x)=e^{ipx}\psi(x)$, $(e^{iqP}\psi)(x)=\psi(x+q)$, and
$$
\big(D(q,p)\psi\big)(x)=e^{-iqp/2}e^{ipx}\psi(x-q).
$$
One can measure the following physically relevant POVMs:
\begin{itemize}
\item Rotated quadrature operators $Q_\theta=(\cos\theta) Q+(\sin\theta) P$ where $\theta\in[0,2\pi)$ so that $Q_0=Q$ and $Q_{\pi/2}=P$. In the position representation, the spectral measure of $Q$ is the canonical spectral measure, $[\Qo(X)\psi](x)=\CHI X(x)\psi(x)$, $X\subseteq\R$ (Borel set), so that
 the spectral measure (rank-1 PVM) of $Q_\theta$ is $\Qo_\theta(X)=R(\theta)\Qo(X)R(\theta)^*$ where $R(\theta)=e^{i\theta N}=\sum_{n=0}^\infty e^{in\theta}\kb n n$ is the (unitary) rotation operator.  Rotated quadaratures can be measured by a balanced homodyne detector where the phase shift $\theta$ is caused by a phase shifter. A single $\Qo_\theta$ cannot be informationally complete (as a PVM) but the whole measurement assemblage $\{\Qo_\theta(X)\}_{\theta\in[0,\pi)\atop X\subseteq\R}$ forms an informationally complete set of effects. Actually, a rank-1 POVM ${\mathsf G}_{\rm ht}(\Theta\times X)=\frac1{\pi}\int_{\Theta}\Qo_\theta(X)\d\theta$) determines the input state completely (optical homodyne tomography, OHT). Note that $\|{\mathsf G}_{\rm ht}(\Theta\times X)\|\le\frac1{\pi}\int_{\Theta}\d\theta<1$ if $\Theta\subseteq[0,\pi)$ is not of `length' $\pi$.
\item The number operator $N=\sum_{n=0}^\infty n\kb n n$ whose spectral measure (rank-1 PVM) is $n\mapsto \mathsf N_n=\kb n n$ (an ideal photon detector with the 100 \% efficiency).
\item An unsharp (rank-$\infty$) number observable (POVM) $$n\mapsto \mathsf N^\epsilon_n=\sum_{m=n}^\infty{m\choose n} \epsilon^n(1-\epsilon)^{m-n}\kb m m$$ (a nonideal photon detector with quantum efficiency $\epsilon\in[0,1)$). Now $\mathsf N^\epsilon$ is neither of norm-1 nor informationally complete, since it is commutative \cite{BuLa89}, and $\lim_{\epsilon\to1}\mathsf N^{\epsilon}_n=\mathsf N_n$ (by $0^0=1$).
\item Covariant phase space observables (POVMs)
$$
\mathsf G_S(Z)=\frac1\pi\int_ZD(z)SD(z)^*\d^2z=\frac1{2\pi}\int_ZD(q,p)SD(q,p)^*\d q\d p,\qquad Z\subseteq\C,
$$
where $S$ is (essentially) the reference state of an eight-port (or double) homodyne detector. In practice $\mathsf G_S$ can be viewed as a joint measurement of unsharp rotated quadratures (e.g., unsharp position and momentum). Moreover, $\mathsf G_S$ is rank-1 if and only if $S=\kb{\psi}{\psi}$ where $\psi$ is a unit vector. Note that $\|\mathsf G_S(Z)\|\le\frac1\pi\int_Z\d^2z<1$ when the area of the set $Z$ is small enough.
\item Covariant phase observables (POVMs)
$$
\mathsf\Phi_C(\Theta)=\frac{1}{2\pi}\int_\Theta R(\theta)CR(\theta)^*\d\theta
=\sum_{n,m=0}^\infty C_{nm}\int_\Theta e^{i(n-m)\theta}\frac{\d\theta}{2\pi}\kb n m
,\qquad \Theta\subseteq[0,2\pi),
$$
where $C=\sum_{n,m=0}^\infty C_{nm}\kb n m$ is a positive sesquilinear form with the unit diagonal ($C_{nn}\equiv1$), i.e.\ a phase matrix.
If $C_{nm}\equiv1$ we get the canonical phase observable $\mathsf\Phi_{\rm can}$,
$$
\mathsf\Phi_{\rm can}(\Theta)=\sum_{n,m=0}^\infty\int_\Theta e^{i(n-m)\theta}\,\frac{\d\theta}{2\pi}|n\>\<m|,\qquad\Theta\subseteq[0,2\pi),
$$
whereas the angle margin of $\mathsf G_S$ is called a phase space phase observable. Both can be measured by double homodyne detection \cite{PeSc}. Any phase observable is never projection valued and is a preprocessed version of the canonical phase given by a (Schur type) quantum channel. Note that $\mathsf\Phi_{\rm can}$ is not informationally complete (consider number states) but it is norm-1. However, $\mathsf\Phi_{\rm can}$ is not pre-processing clean since its nontrivial effects cannot have eigenvalues \cite[Theorem 8.2]{kirja}.
\end{itemize}

\subsection{Definitions and mathematical background}

In this paper, $\N=\{1,\,2,\ldots\}$, i.e.,\ 0 is not included in the set of natural numbers. We define an empty sum $\sum_{j=1}^0(\cdots)$ to be equal to zero. When $\hi$ is a Hilbert space, we denote by $\lh$ the algebra of bounded linear operators on $\hi$ and by $I_\hi$ the unit element of this algebra (the identity operator on $\hi$); by `Hilbert space' we always mean a complex Hilbert space. The inner product of any Hilbert space will be simply denoted by $\<\,\cdot\,|\,\cdot\,\>$ since the Hilbert space in question should always be clear from the context, and the inner product is chosen to be linear in the second argument. By $\cal P(\hi)$, we denote the set of projections of $\hi$, i.e.,\ operators $P\in\lh$ such that $P=P^*=P^2$. An operator $E\in\lh$ is called effect if $0\le E\le I_\hi$ holds. Especially, any projection is an effect, the so-called sharp effect. We let $\th$ stand for the set of trace-class operators on $\hi$, i.e.,\ $\tr{|T|}<\infty$ for all $T\in\th$. We denote the set of positive trace-1 operators in $\th$ by $\sh$; in quantum physics, these normalized positive states of $\lh$ are identified with the physical states of the system described by $\hi$. Note that $\cal P(\hi)\cap\sh$ consists of rank-1 projections $\kb\psi\psi$ (where $\psi\in\hi$ is a unit vector).

When $\mu$ and $\nu$ are positive measures on a measurable space $(\Omega,\Sigma)$ (where $\Omega\neq\emptyset$ is a set and $\Sigma$ is a $\sigma$-algebra of subsets of $\Omega$) we say that $\mu$ is absolutely continuous with respect to $\nu$ and denote $\mu\ll\nu$ if $\mu(X)=0$ whenever $\nu(X)=0$. When both $\mu\ll\nu$ and $\nu\ll\mu$, we denote $\mu\sim\nu$ and say that $\mu$ and $\nu$ are equivalent.

Let $(\Omega,\Sigma)$ be a measurable space and $\hi$ be a Hilbert space. A map $\Mo:\Sigma\to\lh$ is said to be a {\it normalized positive-operator-valued measure (POVM)} if, for all $\rho\in\sh$, the set function $X\mapsto\tr{\rho\Mo(X)}$, denoted hereafter by $p_\rho^\Mo$, is a probability measure or, equivalently, $\Mo(X)\geq0$ for all $X\in\Sigma$, $\Mo(\Omega)=I_\hi$, and, for any pairwise disjoint sequence $X_1,\,X_2,\ldots\in\Sigma$, one has $\Mo(\cup_jX_j)=\sum_j\Mo(X_j)$ (ultra)weakly.
Denote the set of POVMs from $\Sigma$ to $\lh$ by $\O(\Sigma,\hi)$. When $\Po(X)\in\cal P(\hi)$ for all $X\in\Sigma$ for a  POVM $\Po:\Sigma\to\lh$, we say that $\Po$ is a {\it normalized projection-valued measure (PVM)} or a spectral measure. We extend the notions of absolute continuity and equivalence introduced above for scalar measures in the obvious way and thus may write, e.g.,\ for a POVM $\Mo\in\O(\Sigma,\hi)$ and a measure $\mu:\Sigma\to[0,\infty]$, $\Mo\ll\mu$ if $\Mo(X)=0$ whenever $\mu(X)=0$ and, for another POVM $\No:\Sigma\to\lk$, where $\ki$ is some Hilbert space, $\Mo\ll\No$ if $\No(X)=0$ implies $\Mo(X)=0$. We say that $\Mo$ is {\it discrete} if there exist distinct points $\{x_i\}_{i=1}^N\subseteq\Omega$, $N\in\N\cup\{\infty\}$, and effects $\{\Mo_i\}_{i=1}^N\subseteq\lh$ such that $\Mo=\sum_{i=1}^N\Mo_i\delta_{x_i}$ where $\delta_x$ is a Dirac (point) measure concentrated on $x$. Now $\Mo\ll\sum_{i=1}^N\delta_{x_i}$.
A discrete observable $\Mo$ can naturally be identified with the effects $\Mo_i$ and we will use the notation $(\Mo_i)_{i=1}^N$ for $\Mo$ if the outcomes $x\in\Omega$ are not relevant. Note that, if $\hi$ is separable, and $\Mo\in\O(\Sigma,\hi)$, picking any state $\rho\in\sh$ which is faithful, i.e.,\ $\tr{\rho A}=0$ implies $A=0$ for any positive $A\in\lh$ (or, equivalently, the kernel of $\rho$ is $\{0\}$), we have $\Mo\sim p_\rho^\Mo$.

In quantum physics, POVMs are associated in a one-to-one fashion with observables of the system. The observables associated with PVMs are called sharp. In this view, the number $p_\rho^\Mo(X)=\tr{\rho\Mo(X)}\in[0,1]$ is the probability of obtaining a value within the outcome set $X\in\Sigma$ when measuring $\Mo\in\O(\Sigma,\hi)$ and the system being measured is in the quantum state $\rho\in\sh$. 
In realistic physical experiments, we measure only discrete observables which in many cases can be thought as discretizations of continuous observables, i.e.\ for any $\Mo\in\O(\Sigma,\hi)$ one can choose pairwise disjoint sets $X_i\in\Sigma$ whose union is the whole $\Omega$ and define a discrete POVM by $\Mo_i=\Mo(X_i)$ (with the outcome set $\{1,\ldots,N\}$ or $\N$). In this case, one can replace $\Sigma$ with the sub-$\sigma$-algebra generated by the sets $X_i$.

Let $\cal A$ and $\cal B$ be $C^*$-algebras. We say that a linear map $\Phi:\cal A\to\cal B$ is {\it $n$-positive} ($n\in\N$) if the map
$$
\cal M_n(\cal A)\ni(a_{ij})_{i,j=1}^n\mapsto\big(\Phi(a_{ij})\big)_{i,j=1}^n\in\cal M_n(\cal B)
$$
defined between the $n\times n$-matrix algebras over the input and output algebras is positive. If $\Phi$ is $n$-positive for all $n\in\N$, $\Phi$ is said to be {\it completely positive}. Suppose that $\cal A$ and $\cal B$ are unital (with units $1_{\cal A}$ and $1_{\cal B}$) in which case $\Phi$ is called {\it unital} if $\Phi(1_{\cal A})=1_{\cal B}$. For any unital 2-positive map $\Phi:\cal A\to\cal B$ one has the Schwarz inequality, $\Phi(a)^*\Phi(a)\leq\Phi(a^*a)$ for all $a\in\cal A$. We further define ${\bf CP}(\cal A;\hi)$ as the set of completely positive unital linear maps $\Phi:\cal A\to\lh$ whenever $\cal A$ is a unital $C^*$-algebra and $\hi$ is a Hilbert space. Suppose that $\cal A$ and $\cal B$ are von Neumann algebras. We say that a positive map $\Phi:\cal A\to\cal B$ is {\it normal}, if for any increasing (equivalently, decreasing) net $(a_\lambda)_\lambda\subseteq\cal A$ of self-adjoint operators, one has
$$
\sup_\lambda\Phi(a_\lambda)=\Phi\bigg(\sup_\lambda a_\lambda\bigg),
$$
where $\sup b_\lambda$ is the supremum (ultraweak limit) of the increasing net (equivalently, with $\sup$ replaced by $\inf$, the infimum, in the case of a decreasing net).

Fix a $C^*$-algebra $\cal A$ and a Hilbert space $\hi$. For any completely positive map $\Phi:\cal A\to\lh$, there is a Hilbert space $\cal M$, a unital ${}^*$-representation $\pi:\cal A\to\cal L(\cal M)$, and an isometry $J:\hi\to\cal M$ such that $\Phi(a)=J^*\pi(a)J$ for all $a\in\cal A$ and the linear hull of vectors $\pi(a)J\fii$, $a\in\cal A$, $\fii\in\hi$, forms a dense subspace of $\cal M$. Such a triple $(\cal M,\pi,J)$ is called as a {\it minimal Stinespring dilation for $\Phi$} and it is unique up to unitary equivalence, i.e.,\ if $(\cal M',\pi',J')$ is another minimal dilation, then there is a unitary operator $U:\cal M\to\cal M'$ such that $U\pi(a)=\pi'(a)U$ for all $a\in\cal A$ and $UJ=J'$.

Let $\hi$ and $\ki$ be Hilbert spaces. We call normal completely positive maps $\Phi:\lk\to\lh$ satisfying $\Phi(I_\ki)\le I_\hi$ as {\it operations}. When $\Phi$ is in addition unital, i.e.,\ $\Phi(I_\ki)=I_\hi$, we call $\Phi$ as a {\it channel}. For any normal linear map $\Phi:\lk\to\lh$ there exists a (unique) predual map $\Phi_*:\cal T(\hi)\to\cal T(\ki)$ such that
$$
\tr{\Phi_*(T)A}=\tr{T\Phi(A)},\qquad T\in\cal T(\hi),\quad A\in\lk.
$$
The version $\Phi:\lk\to\lh$ is said to be in the {\it Heisenberg picture} and the version $\Phi_*:\cal T(\hi)\to\cal T(\ki)$ is said to be in the {\it Schr\"odinger picture}. For a channel $\Phi$, the Schr\"odinger channel $\Phi_*$, when restricted onto $\cal S(\hi)$, describes how the system associated with $\hi$ transforms under $\Phi$ into another system associated with $\ki$.

Let $(\Omega,\Sigma)$ be a measurable space and $\hi$ and $\ki$ Hilbert spaces. We say that a map $\mathcal J:\lk\times\Sigma\to\lh$ is a (Heisenberg) {\it instrument} if
\begin{itemize}
\item[(i)] $\mathcal J(\cdot,X):\lk\to\lh$ is an operation for all $X\in\Sigma$,
\item[(ii)] $\mathcal J(\cdot,\Omega)$ is a channel, and
\item[(iii)] for any pairwise disjoint sequence $X_1,\,X_2,\ldots\in\Sigma$ and any $A\in\lk$, $\mathcal J(A,\cup_jX_j)=\sum_j\mathcal J(A,X_j)$ (ultra)weakly.
\end{itemize}
For any instrument $\mathcal J:\lk\times\Sigma\to\lh$, we define the predual (Scr\"odinger) instrument $\mathcal J_*:\th\times\Sigma\to\tk$,
$$
\mathcal J_*(T,X)=[\mathcal J(\cdot,X)_*](T),\qquad T\in\th,\quad X\in\Sigma.
$$

Note that, for an instrument $\mathcal J$, the map $\mathcal J(I_\ki,\cdot):\Sigma\to\lh$ is a POVM. On the other hand, for any $\Mo\in\O(\Sigma,\hi)$ and a Hilbert space $\ki$ there is an instrument $\mathcal J:\lk\times\Sigma\to\lh$ such that $\Mo(X)=\mathcal J(I_\ki,X)$ for all $X\in\Sigma$, i.e.,\ $p_\rho^\Mo=\tr{\mathcal J_*(\rho,\cdot)}$; we call $\mathcal J$ an {\it $\Mo$-instrument.} In a measurement of an observable associated with a POVM $\Mo$, the system transforms conditioned by registering an outcome $x\in X$. This conditional state transformation is given by the operation $\mathcal J_*(\cdot,X)$ where $\mathcal J$ is an $\Mo$-instrument. The operator $\mathcal J_*(\rho,X)$ is a subnormalized state whose trace coincides with the probability $p_\rho^\Mo(X)$ of registering an outcome in $X$. If $p_\rho^\Mo(X)>0$ then $[p_\rho^\Mo(X)]^{-1}\mathcal J_*(\rho,X)$ is the corresponding conditional state.

\section{General structure of a quantum observable}

In this section, we analyse the structure of an observable with a general value space on a system described by a separable Hilbert space. We will refer to the results reviewed in this section several times on the course of this paper.

Suppose that $\hi$ is a separable Hilbert space and let $\e=\{h_n\}_{n=1}^{\dim\hi}$ be an orthonormal (ON) basis of $\hi$ and
$$
V_\e:=\lin_\C\{h_n\,|\,1\le n<\dim\hi+1\}.
$$
Note that $V_\e$ is dense in $\hi$. Let $V_\e^\times$ be the algebraic antidual of the vector space $V_\e$, that is, $V_\e^\times$ is the linear space consisting of all antilinear functions $c:\,V_\e\to\C$ (antilinearity means that $c(\alpha\psi+\beta\fii)=\ov\alpha c(\psi)+\ov\beta c(\fii)$ for all $\alpha,\,\beta\in\C$ and $\psi,\,\fii\in V_\e$). By denoting $c_n={c(h_n)}$ one sees that $V_\e^\times$ can be identified with the linear space of formal series $c=\sum_{n=1}^{\dim\hi}c_nh_n$ where $c_n$'s are arbitrary complex numbers. Hence, $V_\e\subseteq\hi\subseteq V_\e^\times$. Denote the dual pairing $\<\psi|c\>:={c(\psi)}=\sum_{n=1}^{\dim\hi} \<\psi|h_n\>c_n$ and $\<c|\psi\>:=\overline{\<\psi|c\>}$ for all $\psi\in V_\e$ and $c\in V_\e^\times$.
Especially, $c_n=\<h_n|c\>$. We say that a mapping $c:\,\Omega\to V_\e^\times,\,x\mapsto \sum_{n=1}^{\dim\hi}c_n(x)h_n$ is {\it (weak${}^*$-)measurable} if its components $x\mapsto c_n(x)$ are measurable \cite{HyPeYl}. Note that, if $c:\,\Omega\to\hi\subseteq V_\e^\times$ is weak${}^*$-measurable then the maps $x\mapsto\<\psi|c(x)\>$ are measurable for all $\psi\in\hi$.

Let $(\Omega,\Sigma)$ be a measurable space and $\hd$ denote a direct integral $\int_\Omega^\oplus\hi(x)\,\d\mu(x)$ of {\it separable} Hilbert spaces $\hi(x)$ such that $\dim\hi(x)=m(x)\in\N\cup\{0,\infty\}$; here $\mu$ is a $\sigma$-finite nonnegative measure\footnote{Note that $\mu$ can be a probability measure everywhere in this paper; any $\sigma$-finite measure is equivalent with a probability measure.} on $(\Omega,\Sigma)$ \cite{Di}. For each $f\in L^\infty(\mu)$, we denote briefly by $\hat f$ the multiplicative (i.e.\ diagonalizable) bounded operator $(\hat f\psi)(x):=f(x)\psi(x)$ on $\hd$. Especially, one has the {\it canonical spectral measure} $\Sigma\ni X\mapsto \Po_\oplus(X):=\CHII X\in\mathcal L(\hd)$ (where $\CHI X$ is the characteristic function of $X\in\Sigma$). We say that an operator $D\in\cal L(\hd)$ is {\it decomposable} if there is a weakly $\mu$-measurable field of operators $\Omega\ni x\mapsto D(x)\in\cal L\big(\hi(x)\big)$ such that $(D\psi)(x)=D(x)\psi(x)$ for all $\psi\in\hd$ and $\mu$-a.a $x\in\Omega$; it is often denoted
$$
D=\int_\Omega^\oplus D(x)\,\d\mu(x).
$$
We have the following theorem proved in \cite{HyPeYl,Pe11}:

\begin{theorem}\label{th1}
Let $\Mo:\,\Sigma\to\lh$ be a POVM and $\mu:\,\Sigma\to[0,\infty]$ a $\sigma$-finite measure such that $\Mo\ll\mu$. Let $\e$ be an ON basis of $\hi$. There exists a direct integral $\hd=\int_\Omega^\oplus\hi(x)\,\d\mu(x)$ (with $m(x)\le\dim\hi$) such that,
for all $X\in\Sigma$,
\begin{enumerate}
\item[(i)]
$\Mo(X)=J_\oplus^*\Po_\oplus(X)J_\oplus$ where $J_\oplus:\,\hi\to\hd$ is a linear isometry such that the set of linear combinations of vectors $\Po_\oplus(X')J_\oplus \fii$, $X'\in\Sigma$, $\fii\in\hi,$ is dense in $\hd$.
\item[(ii)]
There are measurable maps $d_k:\,\Omega\to V_\e^\times$ such that, for all $x\in\Omega$, the vectors $d_k(x)\ne 0$, $k<m(x)+1$, are linearly independent, and
$$
\<\fii|\Mo(X)\psi\>=\int_X \sum_{k=1}^{m(x)} \<\fii|d_k(x)\>\<d_k(x)|\psi\>\,\d\mu(x),\hspace{0.5cm}\fii,\,\psi\in V_\e,
$$
(a minimal diagonalization of $\Mo$). In addition, 
there exist measurable maps $\Om\ni x\mapsto g_\ell(x)\in V_\e$ such that $\<d_k(x)|g_\ell(x)\>=\delta_{k\ell}$ (the Kronecker delta).
\item[(iii)]
$\Mo$ is a spectral measure if and only if $J_\oplus$ is a unitary operator and thus $\hd$ can be identified with $\hi$.
\end{enumerate}
\end{theorem}

A minimal Stinespring dilation for a POVM $\Mo:\Sigma\to\lh$ (viewed as a completely positive map $L^\infty(\mu)\to\lh$, $f\mapsto\int f\,\d\Mo$, where $\mu$ is a probability measure such that $\Mo\ll\mu$) is called as a minimal Na\u{\i}mark dilation and it consists of a Hilbert space $\cal M$, an isometry $J:\hi\to\cal M$, and a spectral measure $\Po:\Sigma\to\cal L(\cal M)$ such that $\Mo(X)=J^*\Po(X)J$ and the vectors $\Po(X)J\fii$, $X\in\Sigma$, $\fii\in\hi$, span a dense subspace of $\cal M$. The above theorem tells that, whenever $\hi$ is separable, one can choose $\cal M=\hd=\int_\Omega^\oplus\hi(x)\,\d\mu(x)$ and $\Po$ to be the canonical spectral measure $\Po_\oplus$.

\subsection{Physical outcome spaces}

It is reasonable to assume that a physically relevant outcome space $(\Omega,\Sigma)$ of an observable is regular or `nice' enough. One can often suppose that $\Sigma$ is {\it countably generated,} i.e.\ there exists a countable $S\subseteq\Sigma$ such that $\Sigma$ is the smallest $\sigma$-algebra of $\Omega$ containing $S$. We will always consider any topological space $T$ as a measurable space $\big(T,\bo T\big)$ where $\bo T$ is the Borel $\sigma$-algebra of $T$. Furthermore, we equip any subset $S$ of $T$ with its subspace topology and the corresponding Borel $\sigma$-algebra $\bo S=\bo T\cap S$. We have the following proposition \cite[Proposition 3.2]{Preston}:
\begin{proposition}
A measurable space $(\Omega,\Sigma)$ is countably generated if and only if there exists a map 
$f:\,\Omega\to\R$ such that
\begin{itemize}
\item[(i)] for all $Y\in\bo\R$ the preimage $f^{-1}(Y)\in\Sigma$ (measurability) and
\item[(ii)] for all $X\in\Sigma$ there is $Y\in\bo\R$ such that $f^{-1}(Y)=X$.
\end{itemize}
\end{proposition}
\noindent
Recall that $f$ satisfying (i) and (ii) is called exactly measurable. If $(\Omega,\Sigma)$ is countably generated and $\mu$ any $\sigma$-finite positive measure on $\Sigma$ then $L^2(\mu)$ and $\hd=\int_\Omega^\oplus\hi(x)\,\d\mu(x)$ are separable.

We say that $(\Omega,\Sigma)$ is {\it nice}\footnote{In \cite{Preston}, nice spaces correspond to  type $\mathcal B$-spaces.} if it is countably generated and  $f:\,\Omega\to\R$ of the above proposition meets the additional condition 
\begin{itemize}
\item[(iii)] 
$f(\Omega)\in\bo\R$.
\end{itemize}
Note that in this case actually $f(X)\in\bo\R$ for all $X\in\Sigma$ \cite[Lemma 4.1]{Preston}. If, in addition, $f$ is injective then the nice space $(\Omega,\Sigma)$ is a {\it standard Borel space} showing that nice spaces are generalizations of standard Borel spaces. Any Borel subset of a separable complete metric space is a standard Borel space and, indeed, any standard Borel space is $\sigma$-isomorphic\footnote{A bijective map between two measurable spaces is a $\sigma$-isomorphism if it is measurable and its inverse is also measurable.} to such a set or even to some compact metric space. Usually in physics, outcome spaces are finite-dimensional second countable Hausdorff manifolds which are (as locally compact spaces) standard Borel.

One can think of nice spaces as standard Borel spaces without the separability property (recall that $\Sigma$ is separable if $\{x\}\in\Sigma$ for all $x\in\Omega$). For any $x\in\Omega$ one can define an atom $A_x:=\bigcap\{X\in\Sigma\,|\,x\in X\}=f^{-1}(\{f(x)\})$ \cite[Lemma 3.1]{Preston} so that a nice space is standard Borel if and only if $A_x=\{x\}$ for all $x\in\Omega$. Hence, atoms of nice spaces may have an `inner structure' (compare to the case of real world atoms).

Suppose that $(\Omega,\Sigma)$ is nice with an $f$ satisfying (i)--(iii). Hence, $f(\Omega)\in\bo\R$ is a standard Borel space and, without restricting generality,\footnote{Since two standard Borel spaces are $\sigma$-isomorphic if and only if they have the same cardinality.} we can assume that
\begin{itemize} 
\item $f(\Omega)=\{1,2,\ldots,N\}$, $N\in\N$, or $f(\Omega)=\N$ (discrete case), or 
\item $f(\Omega)=\R$ (continuous case).
\end{itemize}
In the discrete case, we say that $(\Omega,\Sigma)$ is {\it discrete} and denote $X_i=f^{-1}(\{i\})=A_{x_i}$, $i=1,2,\ldots$, so that $X_i\cap X_j=\emptyset$, $i\ne j$, thus showing that $\Sigma$ is the set of all unions of sets $X_i$ and the empty set $\emptyset$. Moreover, any observable $\Mo:\Sigma\to\lh$ is discrete and, as earlier, can be identified with $(\Mo_i)_{i=1}^N$ where  $\Mo_i=\Mo(X_i)$.

\section{Joint measurability and sequential measurements}

If quantum devices can be applied simultaneously on the same system, we say that they are compatible. Simultaneously measurable observables are called jointly measurable. Let us give formal definitions for these notions.

\begin{definition}
Observables $\Mo_i:\Sigma_i\to\lh$ with outcome spaces $(\Omega_i,\Sigma_i)$, $i=1,2$, are {\it jointly measurable} if they are {\it margins} of a {\it joint observable} $\No:\Sigma_1\otimes\Sigma_2\to\hi$ defined on the product $\sigma$-algebra $\Sigma_1\otimes\Sigma_2$ (generated by sets $X\times Y$, $X\in\Sigma_1$, $Y\in\Sigma_2)$, i.e.,
$$
\Mo_1(X)=\No(X\times\Omega_2),\quad\Mo_2(Y)=\No(\Omega_1\times Y),\qquad X\in\Sigma_1,\quad Y\in\Sigma_2.
$$
\end{definition}
Especially, $\Mo_1$ and $\Mo_2$ are jointly measurable if (and only if) they are {\it functions} or {\it relabelings} of a third observable $\Mo\in\O(\Sigma,\hi)$, i.e.\ for both  $i=1,2$ one has $\Mo_i(X_i)=\Mo\big(f_i^{-1}(X_i)\big)$ for all $X_i\in\Sigma_i$ where $f_i:\,\Omega_i\to\Omega$ is a measurable function. Now a joint observable $\No$ is defined by $\No(X\times Y)=\Mo\big(f_1^{-1}(X)\cap f_2^{-1}(Y)\big)$ for all $X\in\Sigma_1$ and $Y\in\Sigma_2$. Note that, in this case, all the three measurable spaces can be arbitrary \cite[Chapter 11]{kirja}. This implies that, in particular, any observable is jointly measurable with its relabelings.

\begin{definition}
Similarly, we say that an observable $\Mo:\Sigma\to\lh$ and a channel $\Phi:\lk\to\lh$ are {\it compatible} if there exists a {\it joint instrument} $\mathcal J:\lk\times\Sigma\to\lh$ such that
$$
\Mo(X)=\mathcal J(I_\ki,X),\quad\Phi(B)=\mathcal J(B,\Omega),\qquad X\in\Sigma,\quad B\in\lk.
$$
\end{definition}

The above means, when $\Mo$ and $\Phi$ are compatible, there exists a measurement of $\Mo$ such that $\Phi_*$ is the unconditioned state transformation induced by the measurement.

It is useful to look at joint measurablility and compatibility from a more general perspective. Recall the definition of the set ${\bf CP}(\cal A;\hi)$ of unital completely positive maps $\Phi:\cal A\to\lh$. The following result, to which we will often refer, has been obtained, e.g.,\ in \cite{HaHePe14}:

\begin{theorem}\label{theor:tulokuvaus}
Let $\cal A$ and $\cal B$ be von Neumann algebras, $\hi$ a Hilbert space, and $\Psi\in{\bf CP}(\cal A\otimes\cal B;\hi)$. Define the map $\Psi_{(1)}\in{\bf CP}(\cal A;\hi)$, $\Psi_{(1)}(a)=\Psi(a\otimes 1_{\cal B})$ for all $a\in\cal A$, and pick a minimal dilation $(\cal M,\pi,J)$ for $\Psi_{(1)}$. There is a unique map $E\in{\bf CP}(\cal B;\cal M)$ such that
$$
\Psi(a\otimes b)=J^*\pi(a)E(b)J,\qquad \pi(a)E(b)=E(b)\pi(a),\qquad a\in\cal A,\quad b\in\cal B.
$$
If, additionally, $\Psi$ is normal, then both $\pi$ and $E$ are normal.
\end{theorem}

Let us first analyse what the above means for two jointly measurable observables.

\begin{theorem}\label{theor:JMF}
Suppose that $\Mo_i:\Sigma_i\to\lh$, $i=1,2$, are jointly measurable observables on a $\hi$. Let $(\M,\Po,J)$ be any minimal Na\u{\i}mark dilation for $\Mo_1$. Fix a joint observable $\No$ for $\Mo_1$ and $\Mo_2$. There is a unique POVM $\Fo:\Sigma_2\to\cal L(\M)$ such that $\Po(X)\Fo(Y)=\Fo(Y)\Po(X)$ for all $X\in\Sigma_1$ and $Y\in\Sigma_2$ and
$$
\No(X\times Y)=J^*\Po(X)\Fo(Y)J,\qquad X\in\Sigma_1,\quad Y\in\Sigma_2.
$$
\end{theorem}

\subsection{Connection between joint and sequential measurements}

A special case of joint observables is sequential measurements where an initial observable $\Mo:\Sigma\to\lh$ is first measured yielding some $\Mo$-instrument $\mathcal J:\lk\times\Sigma\to\lh$ with output Hilbert space $\ki$. Then some observable $\Mo':\Sigma'\to\lk$ is measured. The conditional probability for obtaining an outcome within $Y\in\Sigma'$ in the second measurement, conditioned by the first measurement observing a value in $X\in\Sigma$, is
$$
\tr{\mathcal J_*(\rho,X)\Mo'(Y)}=\tr{\rho\mathcal J\big(\Mo'(Y),X\big)}
$$
when the system is initially in the state $\rho\in\sh$. For all spaces $(\Omega,\Sigma)$ and $(\Omega',\Sigma')$, the positive operator bimeasure $(X,Y)\mapsto\mathcal J\big(\Mo'(Y),X\big)$ extends into a POVM on $\Sigma\otimes\Sigma'$ \cite{LaYl, ylinen96}. In this case, the extension $\Jo:\Sigma\otimes\Sigma'\to\lh$ is a joint observable for the initial observable $\Mo$ and a distorted version $\Mo''=\mathcal J\big(\Mo'(\,\cdot\,),\Omega\big)$ of the second observable. As shown in Section \ref{sec:intro}, any joint measurement of discrete observables can be implemented as a sequential measurement; see also \cite{HeMi14} for this fact and its generalizations in the case of discrete observables. Next we show that joint and sequential measurements are, in this sense, equivalent in a very general case. Whenever $(\Omega,\Sigma)$ is a measurable space and $\mu$ is a probability measure on $\Sigma$, we denote by $\Po_\mu$ the canonical spectral measure on $L^2(\mu)$, i.e.,\ $\big(\Po_\mu(X)\psi\big)(x)=\CHI X(x)\psi(x)$ for all $X\in\Sigma$, $\psi\in L^2(\mu)$, and $\mu$-a.a.\ $x\in\Omega$. When $\ki$ is a Hilbert space we naturally identify $L^2(\mu)\otimes\ki$ with the $L^2$-space of functions $\Omega\to\ki$.

\begin{proposition}\label{prop:jointtijonoksi}
Suppose that $(\Omega_i,\Sigma_i)$, $i=1,2$, are countably generated measurable spaces and $\hi$ is a separable Hilbert space. Assume that $\Mo_i:\,\Sigma_i\to\lh$, $i=1,2$, are jointly measurable observables with a joint observable $\No$. There is a separable Hilbert space $\ki$, an $\Mo_1$-instrument $\mathcal J:\mathcal L(\ki)\times\Sigma_1\to\lh$, and a POVM $\Mo':\Sigma_2\to\mathcal L(\ki)$ such that
\begin{equation}\label{eq:jointtijonoksi}
\No(X\times Y)=\mathcal J(\Mo'(Y),X),\qquad X\in\Sigma_1,\quad Y\in\Sigma_2.
\end{equation}
\end{proposition}

\begin{proof}
Choose probability measures $\mu_i:\Sigma\to[0,1]$ such that $\Mo_i\ll\mu_i$, $i=1,2$. Pick a minimal Na\u{\i}mark dilation $(\hi_\oplus,\Po_\oplus,J_\oplus)$ of Theorem \ref{th1} for $\Mo_1$ where
$$
\hi_\oplus:=\int_{\Omega_1}^\oplus\hi(x)\,\d\mu_1(x)
$$
is a direct integral space which is separable since $(\Omega_1,\Sigma_1)$ is countably generated.
According to Theorem \ref{theor:JMF}, there is a POVM $\Fo:\Sigma_2\to\mathcal L(\hi_\oplus)$ such that $\Po_\oplus(X)\Fo(Y)=\Fo(Y)\Po_\oplus(X)$ and $\No(X\times Y)=J_\oplus^*\Po_\oplus(X)\Fo(Y)J_\oplus$ for all $X\in\Sigma_1$ and $Y\in\Sigma_2$. Let $(\M,\Qo,K)$ be a minimal Na\u{\i}mark dilation for $\Fo$. Again, $\M$ is separable since $\hd$ is separable and $\Sigma_2$ is countably generated.

Fix $X\in\Sigma_1$ and define $\Fo_X:\Sigma_2\to\mathcal L(\hi_\oplus)$ by $\Fo_X(Y)=\Po_\oplus(X)\Fo(Y)$. Now $\Fo_X(Y)\le \Fo(Y)$ for all $Y\in\Sigma_2$, so that one can define a unique 
$\tilde\Po(X)\in\mathcal L(\M)$ by $\tilde\Po(X)\Qo(Y)K\psi:=\Qo(Y)K\Po_\oplus(X)\psi$, $Y\in\Sigma_2$ and $\psi\in\hd$ (see, e.g., a similar proof of \cite[Proposition 2.1]{HaHePe14}). Clearly,  $\tilde\Po(X)^2=\tilde\Po(X)$, $\tilde\Po(X)\Qo(Y)=\Qo(Y)\tilde\Po(X)$, and  $\Fo_X(Y)=K^*\tilde\Po(X)\Qo(Y)K$ for all $Y\in\Sigma_2$. Hence, $X\mapsto\tilde\Po(X)$ is a PVM.

For all $X\in\Sigma_1$ and $Y\in\Sigma_2$, define the projection $\Ro(X,Y)=\tilde\Po(X)\Qo(Y)\in\li(\M)$. Since $\No$ is a POVM, for any $\fii_1,\,\fii_2\in\hi$, $X_1,\,X_2\in\Sigma_1$, and $Y_1,\,Y_2\in\Sigma_2$, the complex bimeasure
\begin{eqnarray*}
(X,Y)&\mapsto&\<\Qo(Y_1)K\Po_\oplus(X_1)J_\oplus\fii_1|\Ro(X,Y)\Qo(Y_2)K\Po_\oplus(X_2)J_\oplus\fii_2\>\\
&=&\<\fii_1|\No\big((X\times Y)\cap(X_1\times Y_1)\cap(X_2\times Y_2)\big)\fii_2\>.
\end{eqnarray*}
extends into a complex measure on $\Sigma_1\otimes\Sigma_2$. Using the minimality of the subsequent dilations, one finds that $(X,Y)\mapsto\<\xi|\Ro(X,Y)\xi\>$ extends into a measure for all $\xi\in\M$. Thus $(X,Y)\mapsto\tilde\Po(X)\Qo(Y)$ extends into a PVM which we shall also denote by $\Ro$.

Since $\M$ is separable, we may diagonalize $\Ro$ and thus identify $\M$ with the direct integral space
$$
\M_\oplus=\int_{\Omega_1\times\Omega_2}^\oplus\M(x,y)\,\d(\mu_1\times\mu_2)(x,y),
$$
where $\Ro$ operates as the canonical spectral measure. From now on, let us fix a separable infinite-dimensional Hilbert space $\M_\infty$ so that we may define a decomposable isometry $W:\M\to \ov\M:=L^2(\mu_1\times\mu_2)\otimes\M_\infty\cong L^2(\mu_1)\otimes \big[L^2(\mu_2)\otimes\M_\infty\big]$,
$$
W=\int_{\Omega_1\times\Omega_2}^\oplus W(x,y)\,\d(\mu_1\times\mu_2)(x,y),
$$
where $W(x,y):\M(x,y)\to\M_\infty$ are isometries. One may also define the decomposable isometry $W_1:\hi_\oplus\to\ov\M$, $W_1=\int_{\Omega_1}^\oplus W_1(x)\,\d\mu_1(x),$ where $W_1(x):\hi(x)\to L^2(\mu_2)\otimes\M_\infty$ are isometries,
and $\ov K:=WKW_1^*\in\li(\ov\M)$. 

Define the canonical spectral measure $\ov\Ro:=\Po_{\mu_1\times\mu_2}\otimes I_{\M_\infty}$ of $\ov\M$ with the margin $\ov\Po:\Sigma_1\to\mathcal L(\ov\M)$, $\ov\Po(X)=\ov\Ro(X\times\Omega_2)=\Po_{\mu_1}(X)\otimes I_{L^2(\mu_2)}\otimes I_{\M_\infty}$. It is simple to check that $\ov\Ro(Z)W=W\Ro(Z)$ and $\Po_\oplus(X)W_1^*=W_1^*\ov\Po(X)$ for all $Z\in\Sigma_1\otimes\Sigma_2$ and $X\in\Sigma_1$. This means that
$$
\ov\Po(X)\ov K=W\tilde\Po(X)KW_1^*=WK\Po_\oplus(X)W_1^*=\ov K\,\ov\Po(X)
$$
for all $X\in\Sigma_1$. Thus, $\ov K=\int_{\Omega_1}^\oplus\ov K(x)\,\d\mu_1(x)$ where $\ov K(x)\in\mathcal L\big(L^2(\mu_2)\otimes\M_\infty\big)$. Define the isometry $\tilde{K}:=WK=\ov KW_1=\int_{\Omega_1}^\oplus\tilde{K}(x)\,\d\mu(x)$
with the isometries $\tilde{K}(x)=\ov K(x)W_1(x):\,\hi(x)\to L^2(\mu_2)\otimes\M_\infty$.

For $\mu_1$-a.a.\ $x\in\Omega_1$ define the channel 
$$
T_x:\;\mathcal L\big(L^2(\mu_2)\big)\to\mathcal L\big(\hi(x)\big),\quad B\mapsto T_x(B):=\tilde{K}(x)^*(B\otimes I_{\M_\infty})\tilde{K}(x).
$$ 
Since the field $x\mapsto\tilde{K}(x)$ of isometries is measurable, one may define the channel
$$
T:\;\mathcal L\big(L^2(\mu_2)\big)\to\mathcal L(\hi_\oplus),\quad B\mapsto T(B):=\int_\Omega^\oplus T_x(B)\,\d\mu(x).
$$
Using the intertwining properties of the various isometries and POVMs we have, for all $\fii\in\hi$, $X\in\Sigma_1$, and $Y\in\Sigma_2$,
\begin{eqnarray*}
\<J_\oplus\fii|\Po_\oplus(X)T\big(\Po_{\mu_2}(Y)\big)J_\oplus\fii\>&=&\int_X\<(J_\oplus\fii)(x)|T_x\big(\Po_{\mu_2}(Y)\big)(J_\oplus\fii)(x)\>\,\d\mu_1(x)\\
&=&\int_X\<\tilde{K}(x)(J_\oplus\fii)(x)|\big(\Po_{\mu_2}(Y)\otimes I_{\M_\infty}\big)\tilde{K}(x)(J_\oplus\fii)(x)\>\,\d\mu_1(x)\\
&=&\<\tilde{K}J_\oplus\fii|\ov\Ro(X\times Y)\tilde{K}J_\oplus\fii\>=\<KJ_\oplus\fii|\Ro(X\times Y)KJ_\oplus\fii\>\\
&=&\<J_\oplus\fii|\Po_\oplus(X)\Fo(Y)J_\oplus\fii\>=\<\fii|\No(X\times Y)\fii\>.
\end{eqnarray*}
Hence, $\No(X\times Y)=J_\oplus^*\Po_\oplus(X)T\big(\Po_{\mu_2}(Y)\big)J_\oplus$ for all $X\in\Sigma_1$ and $Y\in\Sigma_2$.

Define the instrument $\mathcal J:\;\mathcal L\big(L^2(\mu_2)\big)\times\Sigma_1\to\lh$ by $\mathcal J(B,X)=J_\oplus^*\Po_\oplus(X)T(B)J_\oplus$, see \cite[Theorem 1]{Part2}. The choices $\ki:=L^2(\mu_2)$ and $\Mo':=\Po_{\mu_2}$ yield Equation \eqref{eq:jointtijonoksi}.
\end{proof}

\section{Observables determining the future}

We now turn our attention to those observables which have the property that, no matter how we measure them, registering an outcome unequivocally determines the post-measurement state of the system under study.  Combining Theorem \ref{theor:tulokuvaus} with Theorem \ref{th1}, we obtain the following characterization \cite[Theorem 1]{Part2}:

\begin{theorem}\label{theor:M-compatible}
Let $(\Omega,\Sigma)$ be a measurable space, $\hi$ a separable Hilbert space, and $\Mo:\Sigma\to\lh$ an observable. Pick the minimal Na\u{\i}mark dilation $(\hd,\Po_\oplus,J_\oplus)$ of Theorem \ref{th1} for $\Mo$. Let $\mathcal J:\;\lk\times\Sigma\to\lh$ be an $\Mo$-instrument. There is a unique channel $T:\lk\to\cal L(\hd)$ defined by a (weakly $\mu$-measurable) field $x\mapsto T_x$ of channels $T_x:\lk\to\cal L\big(\hi(x)\big)$,
$$
T(B)=\int_\Omega^\oplus T_x(B)\,\d\mu(x),
$$
i.e.,\ $\big(T(B)\psi\big)(x)=T_x(B)\psi(x)$ for all $B\in\lk$, $\psi\in\hd$, and $\mu$-a.a.\ $x\in\Omega$, such that
$$
\mathcal J(B,X)=J_\oplus^*T(B)\Po_\oplus(X)J_\oplus,\qquad B\in\lk,\quad X\in\Sigma.
$$
\end{theorem}

\begin{definition}
Let an $\Mo\in\O(\Sigma,\hi)$ be associated with the Na\u{\i}mark dilation $(\hd,\Po_\oplus,J_\oplus)$ of Theorem \ref{th1}. If $\dim\hi(x)=1$ for $\mu$-a.a.\ $x\in\Omega$, we say that $\Mo$ is of {\it rank 1}. In this case, $\hd=L^2(\mu)$ and $\Po_\oplus=\Po_\mu$.
\end{definition}

Let the observable $\Mo$ of Theorem \ref{th1} be of rank 1. Also assume that $\mathcal J:\lk\times\Sigma\to\lh$ is an $\Mo$-instrument defined by the pointwise channels $T_x:\lk\to\mathcal L\big(\hi(x)\big)$ of Theorem \ref{theor:M-compatible}. Because of the rank-1 assumption, there are states $\sigma_x\in\cal S(\ki)$ such that $T_x(B)=\tr{\sigma_xB}$, $x\in\Omega$, $B\in\lk$. It follows that $\mathcal J$ is of the following type:

\begin{definition}
Let $\hi$ and $\ki$ be Hilbert spaces and $(\Omega,\Sigma)$ a measurable space. We say that an instrument $\mathcal J:\lk\times\Sigma\to\lh$ is {\it nuclear} if there is a weakly $\mu$-measurable\footnote{In this case, all maps $x\mapsto\tr{\sigma_x B}$, $B\in\lk$, are $\mu$-measurable.} field $\Omega\ni x\mapsto\sigma_x\in\mathcal S(\ki)$  of states such that
$$
\mathcal J(B,X)=\int_X\tr{\sigma_x B}\,\d\Mo(x),\qquad X\in\Sigma,\quad B\in\lk.
$$
\end{definition}

The term {\it nuclear} follows the terminology of Cycon and Hellwig \cite{CyHe}. The above definition means that, in the Schr\"odinger picture, a nuclear instrument $\mathcal J$ has the form
$$
\mathcal J_*(\rho,X)=\int_X\sigma_x\,\d p_\rho^\Mo(x),\qquad\rho\in\sh,\quad X\in\Sigma,
$$
where the integral is defined weakly. Physically this means that a nuclear instrument prepares the quantum system into some post-measurement state which solely depends on the outcome registered, not on the pre-measurement state of the system. This is why also the name {\it measure-and-prepare instrument} could also be used. Thus, any measurement of a rank-1 observable is described by a nuclear instrument and registering a value fully determines the post-measurement state. This is to say, rank-1 observables determine the future of the system under measurement. In fact, also the contrary is true as the following result from \cite{Part2} tells us.

\begin{theorem}\label{theor:rank1-nuclear}
An observable $\Mo:\Sigma\to\lh$ is rank-1 if and only if each $\Mo$-instrument $\mathcal J:\lk\times\Sigma\to\lh$ is nuclear (where $\ki$ is any Hilbert space).
\end{theorem}

The above result can be reformulated in the form that {\it an observable determines the future if and only if it is of rank 1.} The channel $\mathcal J(\cdot,\Omega)$ associated with a measurement of a rank-1 observable is also seen to be entanglement breaking \cite{HoShWe}.

Let $\Mo\in\O(\Sigma,\hi)$ with vectors $d_k(x)$ of Theorem \ref{th1}, $\Om^1:={\N}\times\Om$, and let $\Sigma^1$ be the product $\sigma$-algebra of $2^{\N}$ and $\Sigma$. Let $\mu^1:\,\Sigma^1\to[0,\infty]$ be the product measure of the counting measure and $\mu$. Define $d(k,x)=d_k(x)$ if $k<m(x)+1$ and $d(k,x)=0$ if $k>m(x)$. Then
\begin{equation}\label{eq:M^1}
\<\fii|\Mo^{\bm1}(X^1)\psi\>=\int_{X^1}\<\fii|d(k,x)\>\<d(k,x)|\psi\>\d\mu^1(k,x),\hspace{0.5cm}\fii,\,\psi\in V_\e,\quad X^1\in\Sigma^1,
\end{equation}
defines a rank-1 POVM $\Mo^{\bm1}:\,\Sigma^1\to\lh$; we say that $\Mo^{\bm1}$ is a {\it maximally refined version of $\Mo$}.

Since $\Mo(X)=\Mo^{\bm1}\big(f^{-1}(X)\big)$ where 
$f:\,\Om^1\to\Om$ is a measurable function defined by $f(k,x)=f(x)$ for all $k\in\N$ and $x\in\Om$, $\Mo$ is a relabeling of $\Mo^{\bm1}$. Note that the value space of $\Mo^{\bm1}$ contains  the multiplicities $(k,x)$, $k<m(x)+1$, of a measurement outcome $x$ of $\Mo$. Moreover, $\Mo$ and $\Mo^{\bm1}$ are jointly measurable and $\Mo^{\bm1}$ can be measured by performing a sequential measurement of $\Mo$ and some discrete `multiplicity' observable \cite{Pell'}. We will see that the maximally refined version of an observable possesses many of the same optimality properties as the original observable meaning that we may freely assume the rank-1 property for these observables.

\section{Post-processing and post-processing maximality}

Let us begin with a definition.

\begin{definition}\label{def:kernels}
Let $(\Omega_1,\Sigma_1)$ and $(\Omega_2,\Sigma_2)$ be measurable spaces. Also assume that $\mu:\Sigma_1\to\R$ is a positive measure. We say that a map $\beta:\Sigma_2\times\Omega_1\to\R$ is a {\it $\mu$-weak Markov kernel}  \cite{JePuVi2008} if
\begin{itemize}
\item[(i)] $\beta(Y,\cdot):\Omega_1\to\R$ is $\mu$-measurable for all $Y\in\Sigma_2$,
\item[(ii)] $\beta(Y,x)\geq0$ for all $Y\in\Sigma_2$ and $\mu$-a.a.\ $x\in\Omega_1$,
\item[(iii)] $\beta(\Omega_2,x)=1$ for $\mu$-a.a.\ $x\in\Omega_1$, and
\item[(iv)] for all pairwise disjoint sequences $Y_1,\,Y_2,\ldots\,\in\Sigma_2$,
$$
\beta\big(\cup_{j=1}^\infty Y_j,x\big)=\sum_{j=1}^\infty \beta(Y_j,x)
$$
for $\mu$-a.a.\ $x\in\Omega_1$.
\end{itemize}
If $\beta(\cdot,x)$ is a probability measure for all $x\in\Omega_1$ and the maps $\beta(Y,\cdot)$ are measurable then $\beta$ is simply called a {\it Markov kernel}. 
\end{definition}

When $\mu_1$ is a probability measure on $(\Omega_1,\Sigma_1)$, $\mu_1\ll\mu$, and $\beta:\Sigma_2\times\Omega_1\to\R$ is a $\mu$-weak Markov kernel, then the set function
$$
\Sigma_1\times\Sigma_2\ni (X,Y)\mapsto B(X,Y):=\int_{X}\beta(Y,x)\,\d\mu_1(x)\in[0,1]
$$
is a {\it probability bimeasure}\footnote{Recall that $B:\,\Sigma_1\times\Sigma_2\to\C$ is a bimeasure if $B(X,\,\cdot\,)$, $X\in\Sigma_1$, and $B(\,\cdot\, ,Y)$, $Y\in\Sigma_2$, are (complex) measures.} with the marginal probability measures $X\mapsto B(X,\Omega_2)=\mu_1(X)$ and $Y\mapsto B(\Omega_1,Y)=:\mu_1^\beta(Y)$. As an immediate consequence of Carath\'eodory's extension theorem, one gets the well-know result stating that if $\beta$ is a Markov kernel then $B$ extends into probability measure $\ov B:\,\Sigma_1\times\Sigma_2\to[0,1]$, i.e.,\ $\ov B(X\times Y)=B(X,Y)$ for all $X\in\Sigma_1$ and $Y\in\Sigma_2$. Note that  $\mu_1^\beta$ can be interpreted as a result of (classical) data processing represented by $\beta$. We call this data processing scene {\it post-processing} since the processing can be carried out after obtaining the data represented by the measure $\mu_1$. This data processing scheme generalizes to the case of POVMs in the following way.

\begin{definition}
Let $\Mo_1:\Sigma_1\to\lh$ be an observable operating in the Hilbert space $\hi$. We assume that there is a (probability)measure $\mu$ on $(\Omega,\Sigma)$ such that $\Mo_1\ll\mu$. We say that an observable $\Mo_2:\Sigma_2\to\lh$ is a {\it post-processing of $\Mo_1$}, if there is a $\mu$-weak Markov kernel $\beta:\Sigma_2\times\Omega_1\to\R$ such that $p_\rho^{\Mo_2}=(p_\rho^{\Mo_1})^\beta$ for all $\rho\in\sh$ or, equivalently,
$$
\Mo_2(Y)=\int_{\Omega_1}\beta(Y,x)\,\d\Mo_1(x)\qquad\mathrm{(weakly)}
$$
for all $Y\in\Sigma_2$. We denote $\Mo_2=\Mo_1^\beta$.
\end{definition}

The above means that by measuring $\Mo_1$, we obtain all the information obtainable by measuring $\Mo_2$; we just have to process the data given by $\Mo_1$ classically with the fixed kernel $\beta$. Thus, $\Mo_1$ can give us at least the same amount of  information on the quantum system as $\Mo_2$ modulo classical data processing. Note that if $\Mo_2$ is a relabeling of $\Mo_1$, i.e.\ $\Mo_2(Y)=\Mo_1(f^{-1}(Y))$, then $\Mo_2=\Mo_1^\beta$ where $\beta(Y,x)=\CHI{f^{-1}(Y)}(x)$ is a Markov kernel.

We may thus set up an information-content `order' among observables \cite{BuKeDPeWe2005,DoGr97,MaMu90} $\Mo_2\leq_{\rm post}\Mo_1$ if there is a $\mu$-weak Markov kernel $\beta:\Sigma_2\times\Omega_1\to\R$ (where $\Mo_1\ll\mu$) such that $\Mo_2=\Mo_1^\beta$. We may also say that $\Mo_1$ and $\Mo_2$ are {\it post-processing equivalent} if there are weak Markov kernels $\beta$ and $\gamma$ such that $\Mo_2=\Mo_1^\beta$ and $\Mo_1=\Mo_2^\gamma$. Recall that the `order' $\leq_{\rm post}$ here may not actually be a partial order (because of the failure of transitivity); for situations where this problem can be overcome and identification of canonical representatives of the resulting equivalence classes, see \cite{Kuramochi2015}. An observable $\Mo$ is {\it post-processing maximal} or {\it post-processing clean} if, for any observable $\Mo'$ such that $\Mo\leq_{\rm post}\Mo'$, one has $\Mo'\leq_{\rm post}\Mo$. The maximal observables have been characterized earlier in the case of discrete outcomes \cite[Theorem 3.4]{DoGr97}. We generalize this characterization for observables with nice outcome spaces. For that, we need the following proposition:
 
\begin{proposition}\label{propo3}
Let $(\Omega_1,\Sigma_1)$ be nice, $(\Omega_2,\Sigma_2)$ countably generated, and $B:\,\Sigma_1\times\Sigma_2\to[0,1]$ a probability bimeasure. Denote $\mu_1=B(\,\cdot\, ,\Omega_2)$.
\begin{itemize}
\item[(i)] There exists a probability measure $\ov B:\,\Sigma_1\otimes\Sigma_2\to[0,1]$
such that $\ov B(X\times Y)=B(X,Y)$ for all $X\in\Sigma_1$ and $Y\in\Sigma_2$.
\item[(ii)] There exists a Markov kernel $\beta:\Sigma_2\times\Omega_1\to[0,1]$ such that 
$B(X,Y)=\int_{X}\beta(Y,x)\,\d\mu_1(x)$ for all $X\in\Sigma_1$ and $Y\in\Sigma_2$.
\end{itemize}
\end{proposition}

\begin{proof}
First we note that (i) holds in the case where $(\Omega_1,\Sigma_1)$ and $(\Omega_2,\Sigma_2)$ are standard Borel spaces \cite[Lemma 4.2.1]{Davies} showing that Lemma 12.1 of \cite{Preston} holds even in the case where probability measures on $\Sigma_1\otimes\Sigma_2$ (i.e.\ joint probability measures) are replaced with probability bimeasures on $\Sigma_1\times\Sigma_2$. Manifestly the rest of the proof of Theorem 12.1 of \cite{Preston} can be carried out by replacing joint probability measures with probability bimeasures everywhere. This proves item (ii). Item (i) follows from (ii) by recalling the well-known fact that any Markov kernel defines a joint probability measure.
\end{proof}

\begin{remark}\label{remu}
Let $(\Omega_1,\Sigma_1)$ and $(\Omega_2,\Sigma_2)$ be as in Proposition \ref{propo3}, $\Mo_i\in\O(\Sigma_i,\hi)$, $i=1,2$, $\Mo_1\sim\mu_1$, and $\Mo_2=\Mo_1^\beta$ where $\beta$ is a $\mu_1$-weak Markov kernel, i.e.\ $\Mo_2$ is a post-processing of $\Mo_1$. Since $\beta$ defines a probability bimeasure, we immediately get from Proposition \ref{propo3} the following results:
\begin{itemize}
\item There is a Markov kernel $\beta'$ such that 
$\Mo_2=\Mo_1^{\beta'}$ and $\beta(Y,x)=\beta'(Y,x)$ for all $Y\in\Sigma_2$ and $\mu_1$-a.a.\ $x\in\Omega_1$.
\item The POVMs $\Mo_1$ and $\Mo_2$ are jointly measurable, a joint observable $\No\in\O(\Sigma_1\otimes\Sigma_2,\hi)$ being defined through $\No(X\times Y):=\int_X\beta(Y,x)\d\Mo_1(x).$
\end{itemize}
\end{remark}

\subsection{Joint measurements of rank-1 observables}\label{section5.1}

For the results of the rest of this section, it is useful, as an interlude, to now turn our attention to joint-measurability issues of rank-1 observables. Let $\Mo_i:\,\Sigma_i\to\lh$, $i=1,2$, be jointly measurable observables where $\Mo_1$ is of rank 1. Let $\hi$ be separable and $(\hd,\Po_\oplus,J_\oplus)$ be the minimal (diagonal) Na\u{\i}mark dilation of $\Mo_1$ introduced in Theorem \ref{th1} with the vector field $x\mapsto d_1(x)=:d(x)$ so that, for all $X\in\Sigma_1$,
$$
\<\fii|\Mo_1(X)\psi\>=\int_X\<\fii|d(x)\>\<d(x)|\psi\>\,\d\mu_1(x)=
\int_X \ov{(J_\oplus\fii)(x)}(J_\oplus\psi)(x)\,\d\mu_1(x)
$$
where $\fii,\,\psi\in V_{\bf h}$, since $\hi(x)\equiv\C$ implies $\hd=L^2(\mu_1)$ and $\Po_\oplus=\Po_{\mu_1}$. According to Theorem \ref{theor:JMF}, there is a unique POVM $\Fo:\Sigma_2\to\cal L\big(L^2(\mu_1)\big)$ such that $\Po_{\mu_1}(X)\Fo(Y)=\Fo(Y)\Po_{\mu_1}(X)$ for all $X\in\Sigma_1$ and $Y\in\Sigma_2$ and $\Mo_2(Y)=J_\oplus^*\Fo(Y)J_\oplus$ for all $Y\in\Sigma_2$. Hence, for any $Y\in\Sigma_2$, there is a measurable function $\beta(Y,\,\cdot\,):\Omega_1\to\R$ such that $\big(\Fo(Y)\eta\big)(x)=\beta(Y,x)\eta(x)$ for all $\eta\in L^2(\mu_1)$ and $\mu_1$-a.a.\ $x\in\Omega_1$. It is simple to check that the map $\beta:\Sigma_2\times\Omega_1\to\R$ satisfies the conditions (i)--(iv) of Definition \ref{def:kernels} implying that $\beta$ is a $\mu_1$-weak Markov kernel and $\Mo_2=\Mo_1^\beta$. Thus, we have \cite{Pell2}:

\begin{theorem}\label{theor:rank1JM}
Let $\Mo:\Sigma\to\lh$ be a rank-1 observable of a separable $\hi$. Any observable $\Mo':\Sigma'\to\lh$ jointly measurable with $\Mo$ is a post-processing of $\Mo$.
\end{theorem}

\subsection{Post-processing clean observables}

The general form of post-processing clean observables is claimed to have been solved in  \cite{Beukema2006}. There are, however, some problems in the definition of post-processing the paper uses: Despite the author's definition of post-processing involves, according to the terminology used here, weak Markov kernels, a kernel $\beta$ is treated assuming that $\beta(\cdot,x)$ is a measure for a.a.\ $x$. Moreover, we find the proofs of the main theorems dubious. That is why we provide a new proof. We end up with the same characterization as in \cite{Beukema2006} though. The next theorem is an essential part of the characterization of post-processing clean observables given in Corollary \ref{cor:PostPr<->Rank1}.

\begin{theorem}
Let $(\Omega_i,\Sigma_i)$, $i=1,2$, be measurable spaces, $\hi$ a separable Hilbert space, $\Mo_1\in\O(\Sigma_1,\hi)$, and $\beta:\,\Sigma_2\times\Omega_1\to[0,1]$ a Markov kernel. If $\Mo_1^\beta$ is of rank 1 then $\Mo_1$ is of rank 1.
\end{theorem}

\begin{proof}
Assume that $\mu_1$ is a probability measure on $(\Omega_1,\Sigma_1)$ such that $\Mo_1\ll\mu_1$.
Clearly, $\Mo_2:=\Mo_1^\beta\ll\mu_2:=\mu_1^\beta$ (i.e.\ $\mu_2(Y)=\int_{\Omega_1}\beta(Y,x)\,\d\mu_1(x)$). For any Hilbert-Schmidt operator $R\in\lh$, $Z\in\Sigma_i$, $i=1,\,2$, by the Radon-Nikod\'ym property of the trace class,
$$
R^*\Mo_i(Z)R=\int_X\mo_i(z)\,\d\mu_i(z),
$$
where $\mo_i:\Omega_i\to\lh$ is a weakly $\mu_i$-measurable positive trace-class-valued function (which depends on $R$), see e.g.\ \cite{HyPeYl}. Requiring $\Mo_2$ to be rank-1 is equivalent with $\mo_2(y)$ being at most rank-1 almost everywhere.
Fix now a Hilbert-Schmidt operator $R$ and let $\mo_i$ be the corresponding densities of $R^*\Mo_i(\,\cdot\,)R$ with respect to $\mu_i$. Now
\begin{equation}\label{eq:RM2(Y)R}
R^*\Mo_2(Y)R=\int_{\Omega_1}\beta(Y,x)\mo_1(x)\,\d\mu_1(x)
\end{equation}
for all $Y\in\Sigma_2$. Since $\beta$ is a Markov kernel, the probability bimeasure $(X,Y)\mapsto\int_X\beta(Y,x)\,\d\mu_1(x)$ extends into a probability measure $\mu:\Sigma_1\otimes\Sigma_2\to[0,1]$ whose margins are $\mu_1$ and $\mu_2$. Because $\mu\ll\mu_1\times\mu_2$, there is a (nonnegative) density function $\rho\in L^1(\mu_1\times\mu_2)$ such that
$$
\mu(Z)=\int_Z\rho\,\d(\mu_1\times\mu_2),\qquad Z\in\Sigma_1\otimes\Sigma_2
$$
and, hence, $\int_Y\rho(x,y)\d\mu_2(y)=\beta(Y,x)$ for all $Y\in\Sigma_2$ and $\mu_1$-a.a.\ $x\in\Omega_1$. From Equation \eqref{eq:RM2(Y)R}, it now follows
\begin{equation}\label{m2(y)}
\mo_2(y)=\int_{\Omega_1}\rho(x,y)\mo_1(x)\,\d\mu_1(x)
\end{equation}
for $\mu_2$-a.a.\ $y\in\Omega_2$. Let now, for every $y\in\Omega_2$, $P(y)$ be the at most one-dimensional projection onto the range of $\mo_2(y)$. This is a weakly measurable map. Multiplying \eqref{m2(y)} from both sides with $P(y)^\perp$, one obtains for $\mu_2$-a.a.\ $y\in\Omega_2$
$$
0=\int_{\Omega_1}\rho(x,y)P(y)^\perp\mo_1(x)P(y)^\perp\,\d\mu_1(x).
$$
Thus also
$$
\int_{\Omega_1\times\Omega_2}P(y)^\perp\mo_1(x)P(y)^\perp\,\d\mu(x,y)=\int_{\Omega_1\times\Omega_2}\rho(x,y)P(y)^\perp\mo_1(x)P(y)^\perp\,\d(\mu_1\times\mu_2)(x,y)=0,
$$
implying that $\rho(x,y)P(y)^\perp\mo_1(x)P(y)^\perp=0$ for $(\mu_1\times\mu_2)$-a.a.\ $(x,y)\in\Omega_1\times\Omega_2$.

Denote by $N$ the set of those $(x,y)\in\Omega_1\times\Omega_2$ such that $\rho(x,y)P(y)^\perp\mo_1(x)P(y)^\perp\neq0$. Applying the Fubini theorem for the characteristic function $\CHI N$, one finds that for $\mu_1$-a.a.\ $x\in\Omega_1$, $\rho(x,y)P(y)^\perp\mo_1(x)P(y)^\perp=0$ for $\mu_2$-a.a.\ $y\in\Omega_2$. For $\mu_1$-a.a.\ $x\in\Omega_1$, there is $y\in\Omega_2$ such that $\rho(x,y)>0$. Indeed, if $E\in\Sigma_1$ is such that $\rho(x,y)=0$ for all $x\in E$ and $y\in\Omega_2$, it follows that $0=\int_{E\times\Omega_2}\rho\,\d(\mu_1\times\mu_2)=\mu(E\times\Omega_2)=\mu_1(E)$. Hence, $\mu_1$-a.a.\ $x\in\Omega_1$, there is $y\in\Omega_2$ such that $P(y)^\perp\mo_1(x)P(y)^\perp=0$ implying $\mo_1(x)=P(y)\mo_1(x)P(y)$, i.e.,\ $\mo_1(x)$ is at most rank-1 and, since this holds for any Hilbert-Schmidt operator $R$, we have that $\Mo_1$ is rank-1.
\end{proof}

From Remark \ref{remu} and the theorem above we get:

\begin{corollary}\label{theor:maxtorank1}
Suppose that $(\Omega_1,\Sigma_1)$ (resp.\ $(\Omega_2,\Sigma_2)$) is a nice (resp.\ countably generated) measurable space and $\hi$ is a separable Hilbert space. Let $\Mo_i:\Sigma_i\to\lh$, $i=1,\,2$, be observables such that $\Mo_2$ is of rank 1. If $\Mo_2$ is a post-processing of $\Mo_1$ then $\Mo_1$ is of rank 1 as well.
\end{corollary}

The next corollary gives an exhaustive characterization of post-processing clean observables with a nice value space. Especially, we find that such an observable is post-processing maximal if and only if it determines the future of the system under study.

\begin{corollary}\label{cor:PostPr<->Rank1}
Let $(\Omega,\Sigma)$ be a measurable space, $\hi$ a separable Hilbert space, and $\Mo\in\O(\Sigma,\hi)$. If $\Mo$ is of rank-1 then it is post-processing clean. The converse holds when $(\Omega,\Sigma)$ is nice.
\end{corollary}

\begin{proof}
Suppose first that $\Mo$ is rank-1 and $\mu\sim\Mo$ is a probability measure. Hence, according to Theorem \ref{th1}, $\Mo$ has a minimal Na\u{\i}mark dilation $(L^2(\mu),\Po_\mu,J_\oplus)$. If $\Mo$ is a post-processing of an $\tilde\Mo\in\O(\tilde\Sigma,\hi)$ on some measurable space $(\tilde\Omega,\tilde\Sigma)$, i.e.\ $\Mo=\tilde\Mo^{\tilde\beta}$ where $\tilde\beta$ is a $\tilde\mu$-weak Markov kernel and $\tilde\mu\sim\tilde\Mo$, one can define a positive operator bimeasure 
$$
(X,Y)\mapsto
\int_{Y}\tilde\beta(X,y)\d\tilde\Mo(y)=J_\oplus^*\Po_\mu(X)\Fo(Y)J_\oplus
$$
where now $\Fo$ is of the form $\big(\Fo(Y)\eta\big)(x)=\beta(Y,x)\eta(x)$ for all $Y\in\tilde\Sigma$, $\eta\in L^2(\mu)$ and $\mu$-a.a.\ $x\in\Omega$, and thus $\tilde\Mo=\Mo^{\beta}$, see Section \ref{section5.1} for details.

Assume now that $\Mo$ is post-processing clean. Let $\Mo^{\bm1}:\,\Sigma^1\to\lh$ be the rank-1 refinement of $\Mo$ defined in \eqref{eq:M^1} from which $\Mo$ can be post-processed. Since $\Mo$ is clean, $\Mo^{\bm1}$ is also a post-processing of $\Mo$. If $(\Omega,\Sigma)$ is nice then $(\Omega^1,\Sigma^1)$ is nice (thus countably generated) and Corollary \ref{theor:maxtorank1} implies that $\Mo$ is rank-1 as well (i.e.,\ $\Mo$ and $\Mo^{\bm1}$ coincide).
\end{proof}

\section{Observables determining the past}\label{sec:ObsDetPast}

In this section, we concentrate on observables that define the past of the system under study, i.e.,\ those observables whose measurement outcome statistics completely determine the state of the system prior to the measurement.

Let $\hi$ be a Hilbert space and $(\Omega,\Sigma)$ a measurable space. Let $\Mo\in\O(\Sigma,\hi)$ and recall our earlier definition $p_\rho^\Mo=\tr{\rho\Mo(\,\cdot\,)}$ for all $\rho\in\sh$. Note that the map $\rho\mapsto p_\rho^\Mo$ is an affine map which is continuous with respect to the trace norm on $\sh$ and the total variation norm of probability measures. If this map is an injection, the natural conclusion is that the observable $\Mo$ can separate all states; with different states of the system, the outcome statistics will always differ. How one can actually determine the state of the system prior to the measurement is not discussed here; the reader is redirected to \cite{KiPeSchu2010} for this issue.

This prompts the following definition: an observable $\Mo:\Sigma\to\lh$ is {\it informationally complete} if for $\rho,\,\sigma\in\sh$, $\rho\neq\sigma$ implies $p_\rho^\Mo\neq p_\sigma^\Mo$. This injectivity extends to the whole of $\th$, and thus informational completeness of $\Mo$ is equivalent with the following: for any $T\in\th$, the condition $\tr{T\Mo(X)}=0$ for all $X\in\Sigma$ implies $T=0$. From this we see that the range ${\rm ran}\,\Mo=\{\Mo(X)\,|\,X\in\Sigma\}$ of $\Mo$ has to be extensive enough to separate the trace class $\th$. Indeed, $\Mo$ is informationally complete if and only if the ultraweak closure of the linear hull of ${\rm ran}\,\Mo$ (which coincides with the double commutant $({\rm ran}\,\Mo)''$) is the whole of $\lh$ \cite[Proposition 18.1]{kirja}. We can make the following important immediate observations: {\it If $\Mo$ is informationally complete, its rank-1 refinement $\Mo^{\bm1}$ is informationally complete as well, and any joint measurement of an informationally complete observable with some observable is also informationally complete. 
 More generally, if a post-processing of an observable is informationally complete, then the post-processed observable is also informationally complete.}

To further quantify the informational content of an $\Mo\in\O(\Sigma,\hi)$ in a Hilbert space $\hi$, let us define for each $\rho\in\sh$ the set $[\rho]^\Mo\subseteq\sh$ as the set of those states $\sigma\in\sh$ such that $p_\sigma^\Mo=p_\rho^\Mo$. It is evident that $\Mo$ is informationally complete if and only if $[\rho]^\Mo=\{\rho\}$ for all $\rho\in\sh$. This definition can be generalized to the case of sets $\mathcal O$ of observables (in the same Hilbert space $\hi$):
$$
[\rho]^{\cal O}:=\{\sigma\in\sh\,|\,p_\sigma^{\Mo}=p_\rho^{\Mo},\ \forall\Mo\in\cal O\}
$$
One can say that a set $\cal O$ of observables is informationally complete if $[\rho]^{\cal O}=\{\rho\}$ for all $\rho\in\sh$.

An observable $\Mo:\Sigma\to\lh$ is said to be {\it commutative} if $\Mo(X)\Mo(Y)=\Mo(Y)\Mo(X)$ for all $X,\,Y\in\Sigma$. Let $\cal L\subseteq\lh$ be a set of selfadjoint operators. We call the set of vectors $\fii\in\hi$ such that $L_1\cdots L_n\fii=L_{\pi(1)}\cdots L_{\pi(n)}\fii$ for any $L_1,\ldots,\,L_n\in\cal L$, any permutation $\pi$ of $\{1,\ldots,n\}$, and any $n\in\N$ as the {\it commutation domain of $\cal L$} and denote it by ${\rm com}\,\cal L$. The following results concerning relationships between commutativity and sharpness with informational completeness have been proven in \cite{BuLa89}:

\begin{itemize}
\item Whenever $\dim\hi\geq2$ and $\Mo\in\O(\Sigma,\hi)$ is commutative, $\Mo$ is not informationally complete.
\item A family of mutually commuting spectral measures is never informationally complete.
\item If $\Po:\Sigma\to\lh$ is a spectral measure and $\rho\in\sh$, $[\rho]^\Po=\{\rho\}$ if and only if $\rho$ is pure (a rank-1 projection) and there is an $X\in\Sigma$ such that $\rho=\Po(X)$.
\item If $\cal O$ is an informationally complete set of observables then $\dim{{\rm com}\,\cal L}\leq1$, where $\cal L=\bigcup_{\Mo\in\cal O}{\rm ran}\,\Mo$.
\end{itemize}

The following are examples on informationally complete observables and sets of observables:

\begin{itemize}
\item The set $\{\Qo_\theta\}_{\theta\in[0,\pi)}$ of the rotated quadratures introduced in Section \ref{sec:intro} is informationally complete \cite[Theorem 18.1]{kirja}
\item Equivalently with the above, the homodyne observable $\mathsf G_{\rm ht}:\cal B\big([0,\pi)\times\R\big)\to\cal L\big(L^2(\R)\big)$ defined by $\mathsf G_{\rm ht}(\Theta\times X)=\pi^{-1}\int_\Theta\Qo_\theta(X)\,\d\theta$ is informationally complete.
\item The covariant phase space observable $\mathsf G_S$ introduced in Section \ref{sec:intro} is informationally complete if and only if the support\footnote{That is, the closure of the set of points $(q,p)\in\R^2$ such that $\tr{SD(q,p)}\neq0$.} of the function $(q,p)\mapsto\tr{SD(q,p)}$ is $\R^2$ \cite{KiLaSchuWe2012}.
\end{itemize}

\subsection{Informational completeness within the set of pure states}

Sometimes it is fruitful to consider informational completeness of an observable within a restricted set $\cal P\subseteq\sh$ of states; we are, e.g.,\ already guaranteed that the pre-measurement state $\rho$ is within $\cal P$ and it is enough to only be able to discern between states in $\cal P$ \cite{CaHeSchuTo}. Thus we arrive at {\it informational completeness of an $\Mo\in\O(\Sigma,\hi)$ within $\cal P$} meaning that, whenever $\rho,\,\sigma\in\cal P$, $\rho\neq\sigma$, then $p_\rho^\Mo\neq p_\sigma^\Mo$. When the set $\cal P$ consists of pure states, we identify it with $\{[\fii]\,|\,\fii\in\hi,\ \kb{\fii}{\fii}\in\cal P\}$; here, for any $\fii\in\hi$, we denote $[\fii]:=\{t\fii\,|\,t\in\mathbb T\}$ and $\T:=\{z\in\C\,|\,|z|=1\}$. We get the following result for the case where we have to distinguish a pure state from other pure states:

\begin{proposition}\label{prop:infocompinpure}
Let $(\hi_\oplus,\Po_\oplus,J_\oplus)$ be the minimal Na\u{\i}mark dilation of Theorem \ref{th1} for an $\Mo\in\O(\Sigma,\hi)$ in a separable Hilbert space $\hi$. The observable $\Mo$ is informationally complete within the set $\{\kb\fii\fii\,|\,\fii\in\hi,\,\|\fii\|=1\}$ of pure states if and only if  
$WJ_\oplus\fii\notin J_\oplus\hi$ whenever $\fii\in\hi$ and $W=\int_\Omega^\oplus W(x)\,\d\mu(x)\in\cal L(\hd)$ is a decomposable isometry such that $WJ_\oplus\fii\ne tJ_\oplus\fii$ for all $t\in\mathbb T$. 
\end{proposition}

\begin{proof}
Let $\fii,\,\psi\in\hi$ be unit vectors. We have $p_{\kb{\fii}{\fii}}^\Mo=p_{\kb{\psi}{\psi}}^\Mo$ if and only if $\|(J_\oplus\fii)(x)\|=\|(J_\oplus\psi)(x)\|$ for $\mu$-a.a.\ $x\in\Omega$. One can construct measurable fields $x\mapsto\{e_n(x)\}_{n=1}^{m(x)}\subseteq\hi(x)$, $x\mapsto\{f_n(x)\}_{n=1}^{m(x)}\subseteq\hi(x)$ of orthonormal bases such that
$$
\begin{array}{rcl}
e_1(x)&=&\left\{\begin{array}{ll}
\|(J_\oplus\fii)(x)\|^{-1}(J_\oplus\fii)(x),&(J_\oplus\fii)(x)\neq0\\
\eta(x)&{\rm otherwise}
\end{array}\right.,\\
f_1(x)&=&\left\{\begin{array}{ll}
\|(J_\oplus\psi)(x)\|^{-1}(J_\oplus\psi)(x),&(J_\oplus\psi)(x)\neq0\\
\eta(x)&{\rm otherwise}
\end{array}\right.
\end{array}
$$
for $\mu$-a.a.\ $x\in\Omega$, where $x\mapsto \eta(x)\in\hi(x)$ is a measurable field of unit vectors. Defining $W(x):=\sum_{n=1}^{m(x)}\kb{f_n(x)}{e_n(x)}$ we may set up the decomposable isometry (even unitary) $W=\int_\Omega^\oplus W(x)\,\d\mu(x)$ such that $J_\oplus\psi=WJ_\oplus\fii$ if $p_{\kb{\fii}{\fii}}^\Mo=p_{\kb{\psi}{\psi}}^\Mo$ holds. In reverse, it is simple to check that, whenever $W$ is a decomposable isometry such that $J_\oplus\psi=WJ_\oplus\fii$, then $\|(J_\oplus\fii)(x)\|=\|(J_\oplus\psi)(x)\|$ for $\mu$-a.a.\ $x\in\Omega$, i.e., $p_{\kb{\fii}{\fii}}^\Mo=p_{\kb{\psi}{\psi}}^\Mo$. Thus, $p_{\kb{\fii}{\fii}}^\Mo=p_{\kb{\psi}{\psi}}^\Mo$ if and only if $J_\oplus\psi=WJ_\oplus\fii$ with a decomposable isometry $W\in\cal L(\hd)$. The claim immediately follows from this observation and by noting that $\kb{\fii}{\fii}=\kb{\psi}{\psi}$ if and only if $J_\oplus\psi=tJ_\oplus\fii$ for some $t\in\mathbb T$.
\end{proof}

The above proposition implies the well-known fact stated earlier: a PVM in a separable Hilbert space cannot be informationally complete. In fact such a PVM $\Po$ is not informationally complete even within the set of pure states. Indeed, the isometry $J_\oplus$ in the dilation of Theorem \ref{th1} for $\Po$ is unitary, i.e.,\ $J_\oplus\hi=\hd$.

For another example, as well known, the canonical phase $\mathsf\Phi_{\rm can}$ introduced in Section \ref{sec:intro} is not informationally complete within the set of pure states. To see this using Proposition \ref{prop:infocompinpure}, let us give the minimal Na\u{\i}mark dilation of Theorem \ref{th1} for $\mathsf\Phi_{\rm can}$ in the form $(L^2\big([0,2\pi),(2\pi)^{-1}\d\theta\big),\Po_{\rm can},J_{\rm can})$, where $\Po_{\rm can}$ is the canonical spectral measure of $L^2\big([0,2\pi),(2\pi)^{-1}\d\theta\big)$ and
$$
J_{\rm can}=\sum_{n=0}^\infty\kb{\psi_n}{n},\quad\psi_n(\theta)=e^{-in\theta},\quad0\leq\theta<2\pi,\quad n\in\{0\}\cup\N.
$$
Let $n\in\N$. Since $\psi_n(\theta)=e^{-in\theta}\psi_0(\theta)$ for all $\theta\in[0,2\pi)$, defining the decomposable unitary operator $W_n$  through $(W_n\psi)(\theta)=e^{-in\theta}\psi(\theta)$, $\psi\in L^2\big([0,2\pi),(2\pi)^{-1}\d\theta\big)$, $\theta\in[0,2\pi)$, one has $J_{\rm can}|n\>=\psi_n=W_n\psi_0=W_nJ_{\rm can}|0\>\ne tJ_{\rm can}|0\>$. This proves the claim. 

Let $\Mo\in\O(\Sigma,\hi)$ be an observable in a separable Hilbert space $\hi$ with the minimal Na\u{\i}mark dilation $(\hd,\Po_\oplus,J_\oplus)$ of Theorem \ref{th1}. Let us make a few observations and collect a couple conditions that guarantee that $\Mo$ is not informationally complete within the set of pure states and a necessary and sufficient condition for this.
\begin{itemize}
\item {\it If there exist nonzero vectors $\fii,\,\psi\in\hi$ such that $\<(J_\oplus\fii)(x)|(J_\oplus\psi)(x)\>=0$ for $\mu$-almost all $x\in\Omega$ then $\Mo$ is not informationally complete within the set of pure states.} To see this, fix $\fii$ and $\psi$ satisfying the above condition. Without restricting generality, assume that $\|\fii\|=1=\|\psi\|$. Hence, $\<\fii|\psi\>=\<J_\oplus\fii|J_\oplus\psi\>=0$ and $\fii_\pm=2^{-1/2}(\fii\pm\psi)$ are  unit vectors for which $\fii_+\ne t\fii_-$ for all $t\in\mathbb T$ and $\|(J_\oplus\fii_+)(x)\|=\|(J_\oplus\fii_-)(x)\|$ for $\mu$-a.a.\ $x\in\Omega$, that is, $p_{\kb{\fii_+}{\fii_+}}^\Mo=p_{\kb{\fii_-}{\fii_-}}^\Mo$.

Note, however, that the existence of vectors $\fii$ and $\psi$ of the above condition is not necessary for an informationally incomplete observable, a counterexample being the canonical phase: Assume that $\overline{(J_{\rm can}\fii)(\theta)}(J_{\rm can}\psi)(\theta)=0$ for $\d\theta$-a.a.\ $\theta\in[0,2\pi)$. Then either $J_{\rm can}\psi$ or $J_{\rm can}\fii$ is zero since any Hardy function vanishing on a set of positive measure is 
identically zero. 
\item {\it If there are disjoint sets $X_i\in\Sigma$ and nonzero vectors $\fii_i$ such that $\Mo(X_i)\fii_i=\fii_i$, $i=1,\,2$, then $\Mo$ is not informationally complete within the set of pure states.} Indeed, let $\fii_i\in\hi\setminus\{0\}$ and $X_i$, $i=1,\,2$ be as above. Thus, $\Po_\oplus(X_i)J_\oplus\fii_i=J_\oplus\fii_i$ for all $i=1,2$ implying that
$$
|\<(J_\oplus\fii_1)(x)|(J_\oplus\fii_2)(x)\>|\le\|(J_\oplus\fii_1)(x)\|\,\|(J_\oplus\fii_2)(x)\|=0
$$
for $\mu$-almost all $x\in\Omega$ so that $\Mo$ is not informationally complete within the set of pure states.
\item For any decomposable isometry $W=\int_\Omega^\oplus W(x)\,\d\mu(x)\in\cal L(\hd)$, define the operator $Z_W:=J_\oplus^*WJ_\oplus$. {\it The observable $\Mo$ is informationally complete within the set of pure states if and only if, for any decomposable isometry $W\in\cal L(\hd)$, the operator $Z_W$ strictly decreases the norm (i.e.,\ $\|Z_W\fii\|<\|\fii\|$) for any nonzero $\fii\in\hi$ such that $J_\oplus\fii$ is not an eigenvector of $W$.} (Recall that an isometry may not have any eigenvalues and if eigenvalues exist they belong to $\mathbb T$.) To see this, note that, when $\fii\in\hi$ and $W\in\cal L(\hd)$ is a decomposable isometry, the vector $WJ_\oplus\fii\ne tJ_\oplus\fii$ is not in the subspace $J_\oplus\hi\cong\hi$ if and only if its norm genuinely decreases under the `projection' $J_\oplus^*$, i.e.,\ $\|J_\oplus^*WJ_\oplus\fii\|<\|\fii\|$. Thus we obtain the above as a reformulation of Proposition \ref{prop:infocompinpure}. Note that $\|Z_W\|\le 1$ and, if $WJ_\oplus\fii= tJ_\oplus\fii$, then $Z_W\fii=t\fii$.

Suppose that $w:\Omega\to\mathbb T$ is a $\mu$-measurable function and $W=\int_\Omega^\oplus w(x)I_{\hi(x)}\d\mu(x)$. Denote $Z_w:=Z_W=\int_\Omega w(x)\,\d\Mo(x)$. For the informational completeness of $\Mo$ within the set of pure states, it is necessary that $Z_w$ be strictly norm decreasing in the way defined above for any $\mathbb T$-valued measurable function $w$. We see immediately that this condition becomes also sufficient if $\Mo$ is of rank 1. Note that, if $w$ is not a constant on a set of positive measure, then the corresponding $W$ does not have any eigenvalues. For example, in the case of the canonical phase, 
 $W_n$ ($n\ne 0$) does not have any eigenvalues and
$Z_{W_n}=\int_0^{2\pi}e^{-in\theta}\d\mathsf\Phi_{\rm can}(\theta)=\sum_{m=0}^\infty|m+n\>\<m|$ is an isometry. Often we are interested in the state determination power of the rank-1 refinement of an observable which is why the rank-1 case is of particular importance.
\end{itemize}

Let us take a closer look at a couple of examples utilizing the observations made above.

\begin{example}
Consider a phase space observable $\mathsf G_S$ with some generating positive trace-1 operator $S$. Let us denote the closure of the range of $S$ by $\ki$. The dilation of Theorem \ref{th1} is given by the Hilbert space $L^2(\R^2)\otimes\ki$ identified here with the corresponding $L^2$-space of (equivalence classes of) $\ki$-valued functions, the canonical spectral measure $\Po_\oplus:\,\cal B(\R^2)\to\cal L\big(L^2(\R^2)\otimes\ki\big)$, $\big(\Po_\oplus(Z)\eta)(q,p)=\CHI Z(q,p)\eta(q,p)$ for all $Z\in\cal B(\R^2)$, $\eta\in L^2(\R^2)\otimes\ki$, and a.a.\ $(q,p)\in\R^2$, and the isometry $J_\oplus:L^2(\R)\to L^2(\R^2)\otimes\ki$,
$$
(J_\oplus\fii)(q,p)=\frac{1}{\sqrt{2\pi}}S^{1/2}D(q,p)^*\fii,\qquad\fii\in L^2(\R),\quad (q,p)\in\R^2.
$$
It follows that $\mathsf G_S$ is informationally complete within the set of pure states if and only if, whenever $\R^2\ni(q,p)\mapsto W(q,p)\in\lk$ is a weakly measurable field of isometries, the operator
$$
Z_W=\frac{1}{2\pi}\int_{\R^2}D(q,p)S^{1/2}W(q,p)S^{1/2}D(q,p)^*\d q\d p
$$
strictly decreases the norm of any nonzero vector $\fii$ such that $J_\oplus\fii$ is not an eigenvector of $W$. Especially, if $w:\R^2\to\T$ is measurable then 
$
Z_w=\int_{\R^2}w(q,p)\,\d\mathsf G_S(q,p)
$ 
is strictly norm decreasing in the above sense if $\mathsf G_S$ is informationally complete (within the set of pure states). If $S$ is of rank 1 (i.e.,\ $\mathsf G_S$ is rank-1) then $\mathsf G_S$ is informationally complete within the set of pure states if and only if the operators $Z_w$
are strictly norm decreasing as above.

Let $S=\sum_{i=1}^{{\rm rank}\,S}s_i\kb{\fii_i}{\fii_i}$ be the spectral decomposition of $S$ where $\fii_i\in L^2(\R)$ is a unit eigenvector associated to the eigenvalue $s_i\in(0,1]$ (and $\<\fii_i|\fii_j\>=\delta_{ij}$, $\tr S=\sum_i s_i=1$). Pick a representative $\R\ni x\mapsto \fii_i(x)\in\C$ from each class $\fii_i$ such that $\sum_{i=1}^{{\rm rank}\,S}s_i|\fii_i(x)|^2<\infty$ for all $x\in\R$ and define a positive semidefinite integral kernel $K_S:\R^2\to\C$ by
$$
K_S(x,y):=\sum_{i=1}^{{\rm rank}\,S}s_i\fii_i(x)\overline{\fii_i(y)},\qquad (x,y)\in\R^2.
$$
Indeed, $|K_S(x,y)|^2\le K_S(x,x)K_S(y,y)$ and 
$\int_\R K_S(x,x)\d x=1$ by the monotone convergence theorem. Now $X_{K_S}:=\{x\in\R\,|\,K_S(x,x)\ne 0\}$ is essentially unique in the sense that, if $\tilde{K}_S$ is another integral kernel of $S$ then $X_{K_S}$ and $X_{\tilde{K}_S}$ differ in the set of Lebesgue measure zero. For any $\fii,\,\psi\in L^2(\R)$ and $(q,p)\in\R^2$ one gets
\begin{eqnarray*}
2\pi\<(J_\oplus\fii)(q,p)|(J_\oplus\psi)(q,p)\>&=&\<D(q,p)^*\fii|SD(q,p)^*\psi\> \\
&=&
\int_{X_{K_S}}\int_{X_{K_S}}\overline{(D(q,p)^*\fii)(x)}S(x,y)(D(q,p)^*\psi)(y)\d x\d y \\
&=&
\iint_{Y^q_{S,\fii,\psi}}e^{ipx}\overline{\fii(x+q)}K_S(x,y)
e^{-ipy}\psi(y+q)
\d x\d y
\end{eqnarray*}
where $Y^q_{S,\fii,\psi}=X_{K_S}\times X_{K_S}\cap\{x\,|\,\fii(x+q)\ne 0\}\times \{y\,|\,\psi(y+q)\ne 0\}$. If there exists an $R>0$ such that $X_{K_S}\setminus[-R,R]$ is of measure zero (e.g.\ $S=\kb{\CHI{[0,1]}}{\CHI{[0,1]}}$) then it is easy to find $\fii$ and $\psi$ such that $Y^q_{S,\fii,\psi}$ is zero measurable for all $q\in\R$ and, hence, $\mathsf G_S$ is not informationally complete within the set of pure states by above observation.

To connect our analysis with earlier results, define a continuous square-integrable function 
$$
\hat S:\,\R^2\to\R,\quad(q,p)\mapsto \hat S(q,p):=\tr{D(q,p)S}=e^{iqp/2}\int_\R e^{ipx}K_S(x,x+q)\d x.
$$
If $\hat S$ is integrable, $K_S(x,x+q)=\frac{1}{2\pi}\int_\R e^{-iqp/2} e^{-ipx}\hat S(q,p)\d p$ for all $q\in\R$ and a.a.\ $x\in\R$. If, additionally, $X_{K_S}\setminus[-R,R]$ is of measure zero for some $R>0$, $\hat S(q,p)=0$ for all $p\in\R$ if $|q|>2R$ but $\hat S$ need not be compactly supported (e.g.\ $S=\kb{\CHI{[0,1]}}{\CHI{[0,1]}}$ for which $\hat S(0,p)=i(1-e^{ip})/p$ for all $p\ne0$). Assume then that the support of $\hat S$ is compact and thus contained in a rectangle $[-R_0,R_0]\times[-R_0,R_0]$. Now $\hat S$ is integrable and $S(x,x+q)=0$ for (almost) all $x$ and $q$ such that $|q|>R_0$. Immediately one finds unit vectors $\fii,\,\psi\in L^2(\R)$ such that $\<(J_\oplus\fii)(q,p)|(J_\oplus\psi)(q,p)\>=0$ for all $(q,p)\in\R^2$ thus showing that $\mathsf G_S$ is not informationally complete within the set of pure states. Hence, we have obtained Proposition 20(c) of \cite{CaHeSchuTo} as a special case.
\end{example}

\begin{example}
Let $\Mo:2^{\Omega}\to\cal L(\hi)$ be an observable with an at most countably infinite value space $\Omega=\{x_1,x_2,\ldots\}$, and denote $\Mo_i:=\Mo(\{x_i\})$ for all $i=1,2,\ldots$. Denote the closure of the range of $\Mo_i$ by $\ki_i$ for each $i$ and define the Hilbert space $\ki:=\bigoplus_{i}\ki_i$ which is equipped with the canonical spectral measure $\Po:2^{\Omega}\to\cal L(\ki)$ defined by $\Po(\{x_i\})\bigoplus_{j}\fii_j:=\fii_i$ for all $i$ and all $\bigoplus_{j}\fii_j\in\ki$. Moreover, define the isometry $J:\hi\to\ki$, $\fii\mapsto J\fii=\bigoplus_{i}\Mo_i^{1/2}\fii$. The triple $(\ki,\Po,J)$ is a minimal Na\u{\i}mark dilation for $\Mo$ like the one presented in Theorem \ref{th1}.

The observable $\Mo$ is informationally complete within the set of pure states if and only if, for any isometries $W_i\in\cal L(\ki_i)$ and any $\fii\in\hi$ such that there is no $t\in\mathbb T$ such that $W_i\Mo_i^{1/2}\fii=t\Mo_i^{1/2}\fii$ for all $i$, one has $\|Z_W\fii\|<\|\fii\|$ where $Z_W:=\sum_{i}\Mo_i^{1/2}W_i\Mo_i^{1/2}$. This is a direct consequence of our earlier observations by noting that, when $W_i\in\cal L(\ki_i)$ are isometries and $W:=\bigoplus_{i}W_i$, then $WJ\fii=tJ\fii$ for some $\fii\in\hi$ and $t\in\mathbb T$ if and only if $W_i\Mo_i^{1/2}\fii=t\Mo_i^{1/2}\fii$ for all $i$. If $\Mo$ is of rank 1, this condition can be simplified: $\Mo$ is informationally complete within the set of pure states if and only if, for any (nonconstant) function $ i\mapsto w_i\in\mathbb T$ and any nonzero $\fii\in\hi$ (such that $w_i\Mo_i^{1/2}\fii\ne t\Mo_i^{1/2}\fii$ for all $i$), one has $\|Z_w\fii\|<\|\fii\|$, where $Z_w=\sum_{i}w_i\Mo_i$.

Finally, we note that, if $\{\fii_{ik}\}_{k=1}^{\dim\ki_i}$ is an orthonormal basis of $\ki_i$ for each $i$ one can define (linearly independent) vectors $d_{ik}:=\Mo_i^{1/2}\fii_{ik}$, $k<\dim\ki_i+1$, such that $\Mo_i^{1/2}=\sum_k\kb{\fii_{ik}}{\fii_{ik}}\Mo_i^{1/2}=\sum_k|\fii_{ik}\>\<d_{ik}|$,
$J=\sum_i\sum_k|e_{ik}\>\<d_{ik}|$, and $\Po_i:=\Po(\{x_i\})=\sum_k\kb{e_{ik}}{e_{ik}}$, where  $e_{ik}:=\bigoplus_{j}\delta_{ji}\fii_{ik}$; compare to Section \ref{sec:intro}.
\end{example}

\section{Extreme observables}

The relevant mathematical structures in quantum theory, sets of states, observables, channels, and instruments, are convex. For example, for observables $\Mo,\,\Mo'\in\O(\Sigma,\hi)$ and $p\in[0,1]$, one can determine a mixed observable, a convex combination $p\Mo+(1-p)\Mo'$ by
$$
\big(p\Mo+(1-p)\Mo'\big)(X)=p\Mo(X)+(1-p)\Mo'(X),\qquad X\in\Sigma.
$$
Such mixing of devices can be seen as classical noise produced by an imprecise implementation that produces a measurement of $\Mo$ with relative frequency $p$ and something else otherwise.

An element $x\in K$ in a convex set $K$ set is called {\it extreme} if, for any $y,\,z\in K$ and $p\in(0,1)$, $x=py+(1-p)z$ implies $x=y=z$. Thus extreme quantum devices are free of classical noise due to mixing. The extreme elements of the set of states $\sh$ are the rank-1 projections $\kb{\fii}{\fii}$, $\fii\in\hi$, $\|\fii\|=1$, called as {\it pure states}, whereas  the extreme effects are projections. The general characterizations of extremality for quantum devices follow ultimately from the following result \cite{Ar}:

\begin{theorem}\label{theor:ExtArveson}
Suppose that $\cal A$ is a unital $C^*$-algebra and $\hi$ is a Hilbert space. Let $\Phi\in{\bf CP}(\cal A;\hi)$ and pick a minimal Stinespring dilation $(\cal M,\pi,J)$ for $\Phi$. The map $\Phi$ is an extreme point of the convex set ${\bf CP}(\cal A;\hi)$ if and only if the map
$$
({\rm ran}\,\pi)'\ni D\mapsto J^*DJ\in\lh
$$
defined on the commutant of the range of $\pi$ is an injection.
\end{theorem}

We usually say shortly that an observable $\Mo:\Sigma\to\lh$ is extreme if $\Mo$ is an extreme element of $\O(\Sigma,\hi)$. We may elaborate the above extremality characterization in the case of quantum observables \cite{Pe11}.

\begin{theorem}\label{theor:extobs}
Let $(\hi_\oplus,\Po_\oplus,J_\oplus)$ be the minimal Na\u{\i}mark dilation of Theorem \ref{th1} for an $\Mo\in\O(\Sigma,\hi)$ in a separable Hilbert space $\hi$. The observable $\Mo$ is extreme if and only if, for any decomposable operator $D=\int_\Omega^\oplus D(x)\,\d\mu(x)\in\cal L(\hd)$, the condition $J_\oplus^*DJ_\oplus=0$ implies $D=0$.
\end{theorem}

It is an immediate result of Theorem \ref{theor:extobs} that PVMs are extreme. This can also be proven directly by using the fact that projections are the extreme elements of the set of effects. Also, if $(\Omega,\Sigma)$ is nice and $\dim\hi<\infty$ then an extreme observable $\Mo\in\O(\Sigma,\hi)$ is discrete. Indeed, using an exactly measurable function $f:\Omega\to\R$ such that $f(\Om)\in\cal B(\R)$, we now obtain an extreme observable $\Mo\circ f^{-1}\in\O\big(\cal B(\R),\hi\big)$ which is supported on an at most countable set $\{\lambda_1,\lambda_2,\ldots\}\subset\R$ \cite[Section 5]{HaHePe12}; see also \cite{ChiDASchli}. This means that $\Mo$ is supported by the set $\bigcup_{i}f^{-1}(\{\lambda_i\})$ where $f^{-1}(\{\lambda_i\})$ are atoms of $\Sigma$. Below are some examples on extreme observables which are not PVMs.

\begin{itemize}
\item  One can show that the phase space observable $\mathsf G_S$ introduced in Section \ref{sec:intro} is extreme if and only if $S$ is pure, $S=\kb\psi\psi$, and $(q,p)\mapsto\<\psi|D(q,p)\psi\>\ne 0$ for all $(q,p)\in\R^2$ \cite{HePe12}. Hence, if $\mathsf G_S$ is extreme then it is informationally complete (but the converse does not hold). Especially, when $S=|0\>\<0|$, i.e.,\ the generating state is the vacuum state, then we get the rank-1 informationally complete extreme observable $\mathsf G_{\kb00}$.
\item The canonical phase $\mathsf\Phi_{\rm can}$ introduced in Section \ref{sec:intro} is extreme \cite{HePe09}.
\item Fix a number $m>0$ and denote by $\hat\fii$ the Fourier-Plancherel transformation of $\fii\in L^2(\R)$; if $\fii\in L^1(\R)\cap L^2(\R)$ we may write
$$
\hat\fii(p)=\frac{1}{\sqrt{2\pi}}\int_\R e^{-ipx}\fii(x)\,\d x,\qquad p\in\R.
$$
The canonical time-of-arrival observable $\tau:\cal B(\R)\to\cal L\big(L^2(\R)\big)$ for a free mass-$m$ particle moving in $\R$ defined through
$$
\<\fii|\tau(X)\psi\>=\frac{1}{2\pi m}\int_X\int_0^\infty\int_0^\infty e^{\frac{it}{2m}(p_2^2-p_1^2)}\big(\overline{\hat\fii(p_1)}\hat\psi(p_2)+\overline{\hat\fii(-p_1)}\hat\psi(-p_2)\big)\sqrt{p_1p_2}\,\d p_1\,\d p_2\,\d t
$$
for any $X\in\cal B(\R)$ and any vectors $\fii$ and $\psi$ from the Schwartz space of rapidly decreasing functions, is extreme \cite{HaPe11}.
\end{itemize}

We see from Theorem \ref{theor:extobs} that {\it if $\Mo$ is extreme so is its rank-1 refinement $\Mo^{\bm1}$} \cite{Part2}. The two first examples given above are of rank 1. The third example, however, is of rank 2. The rank-1 refinement of the canonical time observable is $\tau^1:\,2^{\{1,2\}}\otimes\cal B(\R) \to\cal L\big(L^2(\R)\big)$,
$$
\<\fii|\tau^1(\{k\}\times X)\psi\>=\frac{1}{2\pi m}\int_X\int_0^\infty\int_0^\infty e^{\frac{it}{2m}(p_2^2-p_1^2)}\overline{\hat\fii\big((-1)^kp_1\big)}\hat\psi\big((-1)^kp_2\big)\sqrt{p_1p_2}\,\d p_1\,\d p_2\, \d t
$$
for all $X\in\cal B(\R)$, $k=1,\,2$, and all $\fii$ and $\psi$ from the Schwartz space of rapidly decreasing functions. Thus, $\tau^1$ is extreme too.

Let us next consider an example where we show how to construct an observable in a separable Hilbert space with all the optimality properties discussed this far, i.e.,\ a post-processing clean (rank-1) informationally complete extreme observable.

\begin{example}\label{ex:InfocompleteRank1Ext}
 Let $\N_\infty:=\N\cup\{\infty\}$ and $\N_d:=\{1,2,\ldots,d\}$ for each $d\in\N$. Let $\hi_d$, $d\in\N_\infty$, be a $d$-dimensional Hilbert space with $\hi_\infty$ having the orthonormal basis $\{\ket n\}_{n=1}^\infty$. We assume that $\hi_{d_1}\subseteq\hi_{d_2}\subseteq\hi_\infty$ whenever $d_1\leq d_2$ so that, for any $d\in\N$, $\{\ket n\}_{n=1}^d$ is an orthonormal basis for $\hi_d$.

For all $n,\,m\in\N$, pick some numbers $p_{nm}>0$ for which
$$
p:=\sum_{n,m=1}^\infty p_{nm}<\infty.
$$
Let $n,\,m\in\N$ such that $n<m$. Define the following vectors:
$$
f_{nn}:=\ket n,\qquad f_{nm}:=\ket n + \ket m,\qquad f_{mn}:=\ket n - i\ket m,
$$
so that 
\begin{eqnarray*}
\kb{ f_{nm}}{ f_{nm}}&=&\kb n n + \kb n m + \kb m n + \kb m m,\\
\kb{ f_{mn}}{ f_{mn}}&=&\kb n n + i\kb n m - i\kb m n + \kb m m
\end{eqnarray*}
and thus $\kb n n=\kb{f_{nn}}{f_{nn}}$,
\begin{eqnarray*}
2\kb n m &=& \big(\kb{ f_{nm}}{ f_{nm}}-\kb{ f_{nn}}{ f_{nn}}-\kb{ f_{mm}}{ f_{mm}}\big)\\
&&-\,i\big(\kb{ f_{mn}}{ f_{mn}}-\kb{ f_{nn}}{ f_{nn}}-\kb{ f_{mm}}{ f_{mm}}\big),\\
2\kb m n &=& \big(\kb{ f_{nm}}{ f_{nm}}-\kb{ f_{nn}}{ f_{nn}}-\kb{ f_{mm}}{ f_{mm}}\big)\\
&&+\,i\big(\kb{ f_{mn}}{ f_{mn}}-\kb{ f_{nn}}{ f_{nn}}-\kb{ f_{mm}}{ f_{mm}}\big).
\end{eqnarray*}
Hence, for any $d\in\N_\infty$, the linearly independent set\footnote{If $d=\infty$ then the linear span of $\cal B_\infty$ is dense in $\li(\hi_\infty)$ with respect to the weak operator topology.}
$$
\cal B_d:=\big\{p_{nm}\kb{f_{nm}}{ f_{nm}}
\in\li(\hi_d)\,\big|\,n,\,m<d+1 
\big\}
$$
has $d^2$ elements and is a basis of $\li(\hi_d)$.

For each $\cal I\subseteq\N_d\times\N_d$ (or $\cal I\subseteq\N\times\N$ if $d=\infty$) we define a positive trace-class operator
$$
S_{\cal I}:=\sum_{(n,m)\in\cal I} p_{nm}\kb{f_{nm}}{ f_{nm}}\in\li(\hi_d).
$$
Indeed, $\|S_{\cal I}\|\le\tr{|S_{\cal I}|}=
\tr{S_{\cal I}}= \sum_{(n,m)\in\cal I} p_{nm}\|f_{nm}\|^2<2p.$

Let $\cal I\subseteq\N_d\times\N_d$ (or $\cal I\subseteq\N\times\N$ if $d=\infty$), and let
$\cal I_0\subseteq\cal I$ be such that the vectors $f_{nm}\in\hi_d$, $(n,m)\in\cal I_0$, form a basis of the vector space ${\rm lin}\{f_{nm}\,|\,(n,m)\in\cal I\}\subseteq\hi_d$ whose closure is the range of $S_{\cal I}$. Now the maximal number of linearly independent elements of $\{f_{nm}\,|\,(n,m)\in\cal I\}\subseteq\hi_d$ is $\#\cal I_0\le d$. Hence, the rank of $S_{\cal I}$ (and $S_{\cal I_0}\le S_{\cal I}$) is $\#\cal I_0$ and
$$
S_{\cal I}=\sum_{k=1}^{\#\cal I_0} \kb{\fii^{\cal I}_k}{\fii^{\cal I}_k},
$$
where the eigenvectors $\fii_k^{\cal I}\in\hi_d$ form an orthogonal set and the eigenvalues $\|\fii^{\cal I}_k\|^2>0$ are such that $\sum_{k=1}^{\#\cal I_0}\|\fii^{\cal I}_k\|^2=\tr{S_{\cal I}}<\infty$. Note that 
$$
\fii^{\cal I}_k=\|\fii^{\cal I}_k\|^{-2}S_{\cal I}\fii^{\cal I}_k=\|\fii^{\cal I}_k\|^{-2}\sum_{(n,m)\in\cal I} p_{nm}\<f_{nm}|\fii^{\cal I}_k\>f_{nm}\in {\rm ran}\,S_{\cal I}={\rm ran}\,S_{\cal I_0}=\overline{{\rm lin}\{f_{nm}\,|\,(n,m)\in\cal I_0\}}
$$
and $\sum_{k=1}^{\#\cal I_0}\|\fii^{\cal I}_k\|^{-2}\kb{\fii^{\cal I}_k}{\fii^{\cal I}_k}$ is the projection from $\hi_d$ onto the $\#\cal I_0$--dimensional Hilbert space ${\rm ran}\,S_{\cal I_0}$.

We assume\footnote{For example, $\cal I=\N_d\times\N_d$ and $\cal I_0=\{(n,n)\,|\,n\in\N_d\}$ (if $d<\infty$).} next that $\#\cal I_0=d$ so that $\{f_{nm}\,|\,(n,m)\in\cal I_0\}$ is a basis of $\hi_d$ so that $S_{\cal I}$ is of full rank and thus invertible. Let $2^{\cal I}$ be the power set of $\cal I$ and define the following discrete rank-1 POVM $\Mo:\,2^{\cal I}\to\li(\hi_d)$:
$$
\Mo_{nm}:=\kb{\sqrt{p_{nm}} S_{\cal I}^{-1/2} f_{nm}}{\sqrt{p_{nm}} S_{\cal I}^{-1/2} f_{nm}},\qquad (n,m)\in\cal I.
$$
Since, for any complex numbers $c_{nm}$ such that $\sup\{|c_{nm}|\,|\,(n,m)\in\cal I\}<\infty$,
$$
\sum_{(n,m)\in\cal I} c_{nm}\Mo_{nm} =
\sum_{(n,m)\in\cal I} c_{nm}p_{nm}\kb{ S_{\cal I}^{-1/2} f_{nm}}{ S_{\cal I}^{-1/2} f_{nm}}=0
$$
if and only if $\sum_{(n,m)\in\cal I} c_{nm}p_{nm}\kb{ f_{nm}}{ f_{nm}}=S_{\cal I}^{1/2}\sum_{(n,m)\in\cal I} c_{nm}p_{nm}\kb{ S_{\cal I}^{-1/2} f_{nm}}{ S_{\cal I}^{-1/2} f_{nm}}S_{\cal I}^{1/2}=0$ if and only if $c_{nm}\equiv0$, the observable $\Mo$ is extreme with $N=\#\cal I\geq \#\cal I_0=d$ elements. Automatically, $d\leq N \leq d^2$ which must hold for any extreme rank-1 POVM. If $\cal I=\N_d\times\N_d$ (or $\cal I=\N\times\N$ if $d=\infty$) then
$N=d^2$ and $\cal B_d$ is a basis of $\li(\hi_d)$ showing that $\Mo$ is also informationally complete.
Note that, in the case $N=d$, we get a PVM $\Mo$. When $N$ runs from $d$ to $d^2$ the value determination ability weakens but state determination power increases.
\end{example}

\subsection{Joint measurements of extreme observables}

We now concentrate on joint measurements involving extreme observables. To this end, let us recall the sets ${\bf CP}(\cal A;\hi)$ of completely positive unital maps defined on a unital $C^*$-algebra $\cal A$ and taking values in a type-1 factor $\lh$. Let $\cal A$ and $\cal B$ be von Neumann algebras. We say that $\Phi_1\in{\bf CP}(\cal A;\hi)$ and $\Phi_2\in{\bf CP}(\cal B,\hi)$ are {\it compatible}, if there is a {\it joint map} $\Psi\in{\bf CP}(\cal A\otimes\cal B;\hi)$ such that $\Phi_1$ coincides with the margin $\Psi_{(1)}$ and $\Phi_2$ coincides with the margin $\Psi_{(2)}$, where
$$
\Psi_{(1)}(a)=\Psi(a\otimes1_{\cal B}),\quad\Psi_{(2)}(b)=\Psi(1_{\cal A}\otimes b),\qquad a\in\cal A,\quad b\in\cal B.
$$
The following result has been obtained in \cite{HaHePe14}.

\begin{theorem}\label{theor:extmarg}
Let $\cal A$ and $\cal B$ be von Neumann algebras and $\Phi_1\in{\bf CP}(\cal A;\hi)$ and $\Phi_2\in{\bf CP}(\cal B;\hi)$ be compatible.
\begin{itemize}
\item[(i)] If $\Phi_1$ is extreme in ${\bf CP}(\cal A;\hi)$ or $\Phi_2$ is extreme in ${\bf CP}(\cal B;\hi)$ then they have a unique joint map $\Psi\in{\bf CP}(\cal A\otimes\cal B;\hi)$.
\item[(ii)] If both $\Phi_1$ and $\Phi_2$ are extreme then the unique joint map $\Psi$ is extreme in ${\bf CP}(\cal A\otimes\cal B;\hi)$.
\item[(iii)] If $\Phi_1$ or $\Phi_2$ is a *-representation (i.e.,\ especially extreme), one has $\Phi_1(a)\Phi_2(b)=\Phi_2(b)\Phi_1(a)$ for all $a\in\cal A$ and all $b\in\cal B$ and the unique joint map $\Psi\in{\bf CP}(\cal A\otimes\cal B;\hi)$ is given by
$$
\Psi(a\otimes b)=\Phi_1(a)\Phi_2(b),\qquad a\in\cal A,\quad b\in\cal B.
$$
\end{itemize}
\end{theorem}

Applying Theorem \ref{theor:extmarg} to the case of an extreme observable and other measurement devices compatible with this observable, we obtain the following \cite{HaHePe14}:

\begin{theorem}
Let observables $\Mo:\Sigma\to\lh$ and $\Mo':\Sigma'\to\lh$ be jointly measurable, and let $\Phi:\lk\to\lh$ be a channel compatible with $\Mo$. If $\Mo$ is extreme then
\begin{itemize}
\item[(i)] the $\Mo$-instrument $\cal J:\lk\times\Sigma\to\lh$ such that $\cal J(B,\Omega)=\Phi(B)$ for all $B\in\lk$ is unique and
\item[(ii)] the joint observable $\No:\Sigma\otimes\Sigma'\to\lh$ for $\Mo$ and $\Mo'$ is unique.
\end{itemize}
Moreover, if $\Mo$ is a PVM,
\begin{itemize}
\item[(i)'] $\Mo(X)\Phi(B)=\Phi(B)\Mo(X)$ for all $X\in\Sigma$ and all $B\in\lk$ and the only $\Mo$-instrument $\cal J:\lk\times\Sigma\to\lh$ such that $\cal J(\cdot,\Omega)=\Phi$ is given by
$$
\cal J(B,X)=\Phi(B)\Mo(X)=\Mo(X)^{1/2}\Phi(B)\Mo(X)^{1/2},\qquad B\in\lk,\quad X\in\Sigma,
$$
and
\item[(ii)'] $\Mo(X)\Mo'(Y)=\Mo'(Y)\Mo(X)$ for all $X\in\Sigma$ and all $Y\in\Sigma'$ and the only joint observable $\No:\Sigma\otimes\Sigma'\to\lh$ for $\Mo$ and $\Mo'$ is determined by
$$
\No(X\times Y)=\Mo(X)\Mo'(Y)=\Mo(X)^{1/2}\Mo'(Y)\Mo(X)^{1/2},\qquad X\in\Sigma,\quad Y\in\Sigma'.
$$
\end{itemize}
\end{theorem}

The above result essentially means that there is only one way in which an extreme observable can be measured if we fix the unconditioned state transformation associated with the measurement. Similarly, there is only one observable incorporating an extreme marginal observable and some other fixed observable. The corresponding conditions for joint measurements involving PVMs, being from a special subclass of extreme observables, are even more stringent: compatibility or joint measurability with a PVM requires that the other measurement device (observable or channel) commutes with the PVM.

\section{Observables determining their values}

We say that an observable $\Mo:\Sigma\to\lh$ determines its values if for any set of its outcomes we may prepare the system into a state such that, in a measurement of $\Mo$, the values obtained are approximately localized within the given set. Formally, this means that, for any $X\in\Sigma$ such that $\Mo(X)\ne 0$ and any $\varepsilon\in(0,1]$, there is a state $\rho\in\sh$ such that $p_\rho^\Mo(X)>1-\varepsilon$.

Suppose that $\Mo:\Sigma\to\lh$ determines its values and $X\in\Sigma$ is such that $\Mo(X)>0$. We may evaluate for any $\varepsilon\in(0,1]$
$$
\|\Mo(X)\|
=\sup_{\rho\in\sh}\tr{\rho\Mo(X)}=\sup_{\rho\in\sh}p_\rho^\Mo(X)>1-\varepsilon
$$ 
showing that $\|\Mo(X)\|=1$. Indeed, we see that this reasoning can easily be inverted: an observable determines its values if and only if it has the {\it norm-1 property}, i.e.,\ $\|\Mo(X)\|=1$ for all $X\in\Sigma$
such that $\Mo(X)\ne 0$.

A more stringent condition than the norm-1 property is the {\it eigenvalue-1 property}: $\Mo\in\O(\Sigma,\hi)$ is an {\it eigenvalue-1} observable if and only if, whenever $X\in\Sigma$ is such that $\Mo(X)\neq0$, then $\Mo(X)$ has the eigenvalue 1. This means that for any $X\in\Sigma$ such that $\Mo(X)\neq0$ there is a state $\rho_X\in\sh$ ``localized'' in $X$ in the sense that $p_{\rho_X}^\Mo(X)=1$, i.e.,\ the approximate ``$\varepsilon$-localization'' associated with norm-1 observables can be replaced with exact localization. Of course, we may always assume that $\rho_X$ above is pure.
Clearly, PVMs are eigenvalue-1 observables and there exist norm-1 POVMs which have not eigenvalue-1 property, e.g.\ the canonical phase  $\mathsf\Phi_{\rm can}.$

Consider then a norm-1 observable $\Mo\in\O(\Sigma,\hi)$ in a finite-dimensional Hilbert space. Denote $d=\dim\hi$. Since any effect has fully discrete spectrum in finite dimensions, $\Mo$ has the eigenvalue-1 property. For $i=1,\,2$, let the sets $X_i\in\Sigma$ and unit vectors $\fii_i\in\hi$ be such that $\Mo(X_i)\fii_i=\fii_i$ and $X_1\cap X_2=\emptyset$ (we assume that $\Mo$ is not trivial). By using the minimal Na\u{\i}mark dilation for $\Mo$, one gets $\<\fii_1|\fii_2\>=0$. Hence, there exist at most $d$ pairwise disjoint sets $X\in\Sigma$ such that $\Mo(X)\ne 0$. If $(\Omega,\Sigma)$ is, e.g.,\ nice we may conclude that $\Mo$ is discrete, i.e.\ there exist points $x_i\in\Omega$, $i=1,\ldots,N\le d$, so that $\Mo=\sum_{i=1}^N \Mo(A_{x_i})\delta_{x_i}$ where $A_{x_i}$ is the atom associated with $x_i$. This result follows by using an exactly measurable function $f:\Omega\to\R$ such that $f(\Omega)\in\cal B(\R)$ and results of \cite[Section 5]{HaHePe12}.

Finally, let us recall that an eigenvalue-1 observable cannot be informationally complete even within the set of pure states. Indeed, it is easy show that this holds in arbitrary (nonseparable) Hilbert spaces. The question, whether a norm-1 observable can be informationally complete, remains to be answered.

\section{Pre-processing and pre-processing maximality}

Let us start this section with a definition.

\begin{definition}
Let $(\Omega,\Sigma)$ be a measurable space and $\hi$ and $\hi'$ be Hilbert spaces. We say that an observable $\Mo':\Sigma\to\li(\hi')$ is a {\it pre-processing} of an observable $\Mo:\Sigma\to\lh$ if there is a channel $\Phi:\lh\to\li(\hi')$ such that
$\Mo'(X)=\Phi\big(\Mo(X)\big)$ for all $X\in\Sigma.$
\end{definition}

The above definition means that $p_\sigma^{\Mo'}=p_{\Phi_*(\sigma)}^{\Mo}$ for all $\sigma\in\cal S(\hi')$, that is, we may measure $\Mo'$ by first transforming the system with the channel $\Phi$ and then measuring $\Mo$. The predual channel $\Phi_*$ is here seen as a form of quantum pre-processing that is used to process the incoming state carrying quantum information before the measurement of $\Mo$.

Pre-processing gives naturally rise to a partial order within the class of observables with the fixed value space $(\Omega,\Sigma)$ and varying system's Hilbert space $\hi$ \cite{BuKeDPeWe2005}. We denote $\Mo'\leq_{\rm pre}\Mo$ if $\Mo'$ is a pre-processing of $\Mo$ by some channel. We may thus ask which are the pre-processing maximal or clean observables. Maximality of an $\Mo\in\O(\Sigma,\hi)$ means that if $\Mo\leq_{\rm pre}\Mo'$ with some observable $\Mo':\Sigma\to\li(\hi')$ such that  $\Mo'\sim\Mo$ (i.e.\ $\Mo'(X)=0$ exactly when $\Mo(X)=0$) then $\Mo'\leq_{\rm pre}\Mo$. 

In the following subsections, we characterize the pre-processing maximal observables first in the case of discrete outcomes and then for the general case. The reason for this division is that we may formulate the maximality in a tighter fashion for discrete observables. Let us first recall a result from \cite{Pe11}:

\begin{theorem}\label{theor:JP-ExtPure}
Suppose that $\hi$ and $\hi'$ are separable Hilbert spaces, $\Po:\Sigma\to\li(\hi')$ is a sharp observable (a PVM), and $\Mo:\Sigma\to\lh$ is some observable such that $\Mo\ll\Po$. There exists a channel $\Phi':\li(\hi')\to\lh$ such that $\Mo(X)=\Phi'\big(\Po(X)\big)$ for all $X\in\Sigma$, i.e.,\ $\Mo\leq_{\rm pre}\Po$.
\end{theorem}

\subsection{Case of discrete observables}

The theorem below characterizes pre-processing clean discrete observables.

\begin{theorem}\label{theor:eigenval1}
Suppose that $\Omega$ is a finite or a countably infinite set and $\Mo:2^{\Omega}\to\lh$ is an observable in a separable Hilbert space $\hi$. Then $\Mo$ is pre-processing clean if and only if it has the eigenvalue-1 property.
\end{theorem}

\begin{proof}
We give the proof for the `only if' part for an observable with a more general value space $(\Omega,\Sigma)$ where $\Sigma$ is countably generated, since this yields no extra complications. Assume that $\Mo\in\O(\Sigma,\hi)$ is pre-processing clean and $\mu$ is a probability measure on $(\Omega,\Sigma)$ such that $\Mo\sim\mu$. Also define the PVM $\Po_\mu:\Sigma\to\li\big(L^2(\mu)\big)$, $\big(\Po_\mu(X)\psi\big)(x)=\CHI X(x)\psi(x)$ for all $X\in\Sigma$, all $\psi\in L^2(\mu)$, and $\mu$-a.a.\ $x\in\Omega$. Hence, since $L^2(\mu)$ is separable, according to Theorem \ref{theor:JP-ExtPure}, $\Po_\mu(X)=\Phi\big(\Mo(X)\big)$ where $\Phi:\lh\to\cal L\big(L^2(\mu)\big)$ is a channel and thus, for all $\rho\in\mathcal S\big(L^2(\mu)\big)$, $X\in\Sigma$, one has $\tr{\rho\Po_\mu(X)}=\tr{\Phi_*(\rho)\Mo(X)}$. We define $\rho_X:=\mu(X)^{-2}\kb{\CHI X}{\CHI X}\in\mathcal S\big(L^2(\mu)\big)$ when $\mu(X)>0$ (or $\Mo(X)\ne 0$). Let $\Phi_*(\rho_X)=\sum_{n=1}^{r}\lambda_n\kb{\fii_n}{\fii_n}$, $\lambda_n>0$, $\sum_{n=1}^r\lambda_n=1$, $\<\fii_n|\fii_m\>=\delta_{nm}$, be the spectral decomposition of the state $\Phi_*(\rho_X)$. Now $1=\tr{\rho_X\Po_\mu(X)}=\tr{\Phi_*(\rho_X)\Mo(X)}=\sum_{n=1}^r\lambda_n\<\fii_n|\Mo(X)\fii_n\>$ implying that $\<\fii_n|\Mo(X)\fii_n\>=1$ and thus $\|\sqrt{\Mo(\Om\setminus X)}\fii_n\|^2=\<\fii_n|\Mo(\Om\setminus X)\fii_n\>=\<\fii_n|[I_\hi-\Mo(X)]\fii_n\>=0$ or 
$\Mo(\Om\setminus X)\fii_n=\sqrt{\Mo(\Om\setminus X)}\sqrt{\Mo(\Om\setminus X)}\fii_n=0$ or
$\Mo(X)\fii_n=1\cdot\fii_n$ for all $n<r+1$.

Let us prove the `if' part for a discrete observable $\Mo$ which we identify with the effects $\Mo_i=\Mo\big(\{x_i\}\big)\neq0$ such that $\sum_{i=1}^{N}\Mo_i=I_\hi$, $N\in\N_\infty$ (without restricting generality, we assume that $\Omega=\{x_i\}_{i=1}^N$). Suppose then that any $\Mo(X)$, $X\ne\emptyset$, has the eigenvalue 1; denote the projection onto the corresponding eigenspace by $P_X$ so that $P_X\Mo(X)=P_X=\Mo(X)P_X$. If $\emptyset\ne X\subseteq Y$ then $\<\psi|\Mo(X)\psi\>\leq \<\psi|\Mo(Y)\psi\>$ for all $\psi\in\hi$ and $\Mo(X)\fii=\fii$ implies $\Mo(Y)\fii=\fii$, that is, $P_X\leq P_Y$. Similarly, if $X\cap Y=\emptyset$, then $P_XP_Y=0$. Let us pick, for all $i<N+1$, $\fii_i\in P_{\{x_i\}}\hi$, $\|\fii_i\|=1$, and define the projection $R:=\sum_{i=1}^N\kb{\fii_i}{\fii_i}$. Using the above results, one immediately sees that $R\Mo_iR=\kb{\fii_i}{\fii_i}$. Define the channel $\Psi:\lh\to\cal L(R\hi)$, $\Psi(A)=RAR$. Now $\Psi$ pre-processes $\Mo$ into the sharp observable $\Po:2^{\Omega}\to\cal L(R\hi)$ defined via $\Po(\{x_i\}):=\kb{\fii_i}{\fii_i}$. If $\Mo$ was a pre-processing of another observable $\Mo'$ on $2^\Omega$ (automatically $\Mo'\sim\Mo\sim\Po$) then $\Po$ would also be a pre-processing of $\Mo'$. According to Theorem \ref{theor:JP-ExtPure}, $\Mo'\leq_{\rm pre}\Po$ and, thus, $\Mo'\leq_{\rm pre}\Mo$.
\end{proof}

Note that the first part of the proof above shows that, even in the case where the value space $(\Omega,\Sigma)$ of a pre-processing maximal observable $\Mo$ is countably generated, $\Mo$ necessarily possesses the eigenvalue-1 property. This tells us that pre-processing maximal observables (with countably generated value spaces) determine their values and are not informationally complete.

\subsection{Case of general observables}\label{sec:preprcont}

For the characterization of pre-processing clean observables with more general (countably generated) value spaces, we need first some auxiliary results. The proof of the following lemma follows closely the one of \cite[Lemma 8.1]{kirja}.

\begin{lemma}\label{lemma:comm}
Let $A\in\lh$, $0\leq A\leq I_\hi$, and $R\in\ph$, where $\hi$ is a Hilbert space. If $RAR\in\ph$ then $A$ and $R$ commute.
\end{lemma}

\begin{proof}
Suppose that $RAR$ is a projection. We now have
$$
\big((I-R)AR\big)^*\big((I-R)AR\big)=RA(I-R)AR=0
$$
implying that $\big((I-R)AR\big)=0$, i.e.,\ $AR=RAR$. We obtain
$$
AR=RAR=(RAR)^*=(AR)^*=RA,
$$
proving the claim.
\end{proof}

Let $\hi$ and $\ki$ be Hilbert spaces and $\Phi:\lh\to\lk$ a channel. There is a minimal projection $R\in\ph$ such that $\Phi(R)=I_\ki$, i.e.,\ if $Q\in\ph$ is such that $\Phi(Q)=I_\ki$ and $Q\leq R$, then $Q=R$. Thus, $R$ can be called the support of $\Phi$ (indeed it is uniquely defined by $\Phi$). Moreover, $\Phi(A)=\Phi(RAR)$ for all $A\in\lh$ and, whenever $A\in\lh$ is an effect, $\Phi(A)=0$ if and only if $RAR=0$ \cite[Section 10.8]{kirja}. We now easily obtain the following result.

\begin{lemma}\label{lemma:FiiR}
Let $\hi$ and $\ki$ be Hilbert spaces and $\Phi:\lh\to\lk$ a channel with the support projection $R$. If $\Phi(A)$ is a projection for some effect $A\in\lh$, then $RAR$ is a projection, $RA=AR$, and $A=RAR+ R^\perp AR^\perp.$
\end{lemma}

\begin{proof}
Let $A\in\lh$, $0\leq A\leq I_\hi$ and assume that $\Phi(RAR)$ is a projection. Using the Schwartz inequality (applicable especially to unital completely positive maps),
$$
\Phi(RAR)=\Phi(RAR)^2\leq\Phi(RARAR)
$$
implying $\Phi(RAR-RARAR)\leq0$. Since $RAR\geq RARAR$, one finds that $\Phi(RAR-RARAR)=0$. Since $RAR-RARAR=R(A-ARA)R$ and $0\leq A-ARA\leq I_\hi$, the properties of the support projection cited above imply that $RAR=RARAR$, i.e.,\ $RAR\in\ph$.
Lemma \ref{lemma:comm} yields $AR=RA$ so that $A=RAR+RAR^\perp+R^\perp AR+R^\perp AR^\perp=RAR+ R^\perp AR^\perp$.
\end{proof}

\begin{theorem}\label{cor:PrePr<->suorsum}
Suppose that the $\sigma$-algebra $\Sigma\subseteq 2^\Omega$ is countably generated and $\Mo:\Sigma\to\lh$ is an observable operating in a separable Hilbert space $\hi$. Then $\Mo$ is pre-processing clean if and only if there exist a closed subspace $\mathcal M\subseteq\hi$, a  PVM $\sfe:\Sigma\to\li(\mathcal M)$, and a  POVM $\Fo:\Sigma\to\mathcal L(\mathcal M^\perp)$ such that $\Mo\sim\sfe$ and
\begin{equation}\label{eq:suorsum}
\Mo(X)=\sfe(X)\oplus\Fo(X),\qquad X\in\Sigma.
\end{equation}
\end{theorem}

\begin{proof}
Suppose that $\Mo$ is pre-processing clean. Then $\Mo$ can be pre-processed with a channel $\Phi:\lh\to\cal L\big(L^2(\mu)\big)$ into the observable $\Po_\mu$ of the first part of the proof of Theorem \ref{theor:eigenval1}, where $\mu\sim\Mo$. We have by Lemma \ref{lemma:FiiR} that there is an $R\in\ph$ such that $\Phi(R)=I_{L^2(\mu)}$ and $\Mo(X)=R\Mo(X)R+R^\perp\Mo(X)R^\perp$, where $R\Mo(X)R$ is a projection, for all $X\in\Sigma$. Suppose that $R\Mo(X)R=0$ for some $X\in\Sigma$. Now $\Po_\mu(X)=\Phi(R\Mo(X)R)=0$ implying that $\mu(X)=0$. Hence $\Mo\ll R\Mo(\,\cdot\,)R$; the contrary is, of course, automatically satisfied. Thus, we may choose $\mathcal M=R\hi$ and $\sfe(X)=R\Mo(X)R$ and $\Fo(X)=R^\perp\Mo(X)R^\perp$ for all $X\in\Sigma$.

Assume then that the decomposition of $\Mo$ into $\sfe$ and $\Fo$ of the claim exists. Clearly, $\sfe\leq_{\rm pre}\Mo$ (with a rank-1 channel defined by the projection from $\hi$ onto $\mathcal M$). 
If $\Mo':\Sigma\to\li(\hi')$ is an observable such that $\Mo'\sim\Mo\sim\sfe$ and $\Mo\leq_{\rm pre}\Mo'$ then, according to Theorem \ref{theor:JP-ExtPure}, $\Mo'\leq_{\rm pre}\sfe\leq_{\rm pre}\Mo$.
\end{proof}

As also seen in the proof of Theorem \ref{theor:eigenval1}, for $\Mo\in\O(\Sigma,\hi)$ to be pre-processing clean, all the nonzero effects $\Mo(X)$ must have the eigenvalue 1. This is clearly satisfied by the direct sum form \eqref{eq:suorsum} since $\sfe$ is a PVM and, whenever $\sfe(X)=0$, $\Mo(X)=0$ as well. However, the contrary is problematic: if all nonzero $\Mo(X)$ have the eigenvalue 1, does it follow that we have the decomposition \eqref{eq:suorsum} with a fixed subspace $\cal M$? This would mean that Theorem \ref{theor:eigenval1} extends plainly to general observables. We leave this as an open question.
Moreover, in the finite-dimensional case, norm-1 observables also have the eigenvalue-1 property, which for (discrete) POVMs implies post-processing maximality; recall that now norm-1 POVMs are discrete if their outcome spaces are regular enough. An important infinite-dimensional and continuous counter example is the canonical phase observable $\mathsf\Phi_{\rm can}$ which is of norm 1, but none of its effects $\mathsf\Phi_{\rm can}(\Theta)\ne I_\hi$ has the eigenvalue 1. Especially, $\mathsf\Phi_{\rm can}$ is not pre-processing maximal.

A different analysis of pre-processing can result in remarkably different characterizations of pre-processing clean observables. For instance, the authors of \cite{BuKeDPeWe2005} concentrate on finite-outcome observables on a {\it fixed} finite-dimensional Hilbert space $\hi$. In this setting, an $N$-valued observable $\Mo$ in $\hi$ is {\it clean} if for any $N$-valued observable $\Mo'$ in $\hi$ such that there exists a channel $\Phi:\lh\to\lh$ with $\Mo_i=\Phi(\Mo'_i)$, $i=1,\ldots,\,N$, there also exists a channel $\Psi:\lh\to\lh$ such that $\Mo'_i=\Psi(\Mo_i)$, $i=1,\ldots,\,N$. With this definition, the set of clean observables within the set of $N$-valued observables in $\hi$ are exactly those $\Mo$ such that $\|\Mo_i\|=1$ for all $i=1,\ldots,\,N$, in the case where $N\leq\dim\hi$. However, now also the case $N>\dim\hi$ is possible and, in general, any rank-1 observable is clean. The difference in the definition of post-processing and clean observables of \cite{BuKeDPeWe2005} and the corresponding definitions of this paper is that, in \cite{BuKeDPeWe2005} one is restricted to using a single system within which to carry out pre-processing whereas in our analysis one is free to use any systems for pre-processing (no limitation to dimensionality of the Hilbert space from which one pre-processes).

\begin{remark}\rm
Norm-1 observables, and hence also pre-processing maximal observables as eigenvalue-1 observables, are an example of so-called regular observables \cite[Section 11.3]{kirja}: An effect $E\in\lh$ is called {\it regular} if $E\not\leq I_\hi-E$ and $I_\hi-E\not\leq E$ or, equivalently, the spectrum of $E$ extends both above and below $1/2$. 
For example, a rank-1 effect $p\kb{\fii}{\fii}$, $p\in(0,1]$, $\fii\in\hi$, $\|\fii\|=1$, is regular if and only if $p>\frac12$ (if $\dim\hi>1$).
An observable $\Mo:\Sigma\to\lh$ is {\it regular} if $\Mo(X)$ is regular whenever $0\neq\Mo(X)\neq I_\hi$. There exist regular POVMs which are not of norm-1 (e.g.\ $\Mo_1=\frac13\kb11+\frac23\kb22$, $\Mo_2=\frac23\kb11+\frac13\kb22$ is regular but not norm-1).

It is simple to check that whenever $\Mo:\Sigma\to\lh$ is regular, its range ${\rm ran}\,\Mo$ equipped with the intersection $\wedge_{{\rm ran}\,\Mo}$,
$$
\Mo(X)\wedge_{{\rm ran}\,\Mo}\Mo(Y):=\inf\{\Mo(Z)\,|\,\Mo(Z)\leq\Mo(X),\,\Mo(Y)\},\qquad X,\,Y\in\Sigma,
$$
and the complementation $':\Mo(X)\mapsto\Mo(X)'=I_\hi-\Mo(X)$ is a Boolean algebra, i.e.,\ especially $\Mo(X)\wedge_{{\rm ran}\,\Mo}\Mo(X)'=0$ for all $X\in\Sigma$. In fact, the converse is true as well \cite{DvPu94}: if $\Mo:\Sigma\to\lh$ is an observable such that $({\rm ran}\,\Mo,\wedge_{{\rm ran}\,\Mo},')$ described above is a Boolean algebra, then $\Mo$ is regular. Hence, a regular POVM $\Mo$ preserves the `classical' Boolean logic between the Boolean algebras $\Sigma$ and ${\rm ran}\,\Mo$.

Whether a regular observable can be informationally complete remains to be seen. This is not possible in the finite-dimensional case. To see this, let us consider an $N$-valued observable $\Mo=(\Mo_i)_{i=1}^N$ in a $d$-dimensional ($d<\infty$) Hilbert space $\hi$. Taking the trace on both sides of the equation $I_\hi=\sum_{i=1}^N\Mo_i$ and assuming that $\Mo$ is informationally complete and regular, one arrives at $d=\sum_{i=1}^N\tr{\Mo_i}>N/2\geq d^2/2$, where the first inequality follows from regularity and the second from informational completeness. This is possible only if $d=1$.

The above no-go result can be alleviated by relaxing the requirement on informational completeness within the set of {\it all} states. Let us consider an example in the two-dimensional Hilbert space. Fix the Pauli matrices $\sigma_x$, $\sigma_y$, and $\sigma_z$ which in the eigenbasis of $\sigma_z$ take the form
$$
\sigma_x=\left(\begin{array}{cc}0&1\\1&0\end{array}\right),\quad\sigma_y=\left(\begin{array}{cc}0&-i\\i&0\end{array}\right),\quad\sigma_z=\left(\begin{array}{cc}1&0\\0&-1\end{array}\right).
$$
Also denote ${\bf b}\cdot\boldsymbol{\sigma}:=b_x\sigma_x+b_y\sigma_y+b_z\sigma_z$ for any ${\bf b}=(b_x,b_y,b_z)\in\R^3$. We now define the three-valued observable $\Mo=(\Mo_1,\Mo_2,\Mo_3)$ by $\Mo_i=3^{-1}(I+{\bf a}_i\cdot\boldsymbol\sigma)$, $i=1,\,2,\,3$, where
$$
{\bf a}_1=(1,0,0),\quad{\bf a}_2=\Big(-\frac{1}{2},\frac{\sqrt{3}}{2},0\Big),\quad{\bf a}_3=\Big(-\frac{1}{2},-\frac{\sqrt{3}}{2},0\Big).
$$
It is easily checked that $\Mo$ is an extreme rank-1 observable and the non-zero eigenvalue of $\Mo_i$ is 2/3 for each $i=1,\,2,\,3$. Thus, $\Mo$ is regular. Moreover, $\Mo$ is informationally complete within the restricted set of states $\rho$ such that $\tr{\rho\sigma_z}=0$. Thus, $\Mo$ is `more informationally complete' than any PVM can be in the two-dimensional case. Indeed, whenever $\Po=(\kb{d_1}{d_1},\kb{d_2}{d_2})$ is a PVM, the maximal subset of states where $\Po$ is informationally complete is the convex hull of the states $\kb{d_1}{d_1}$ and $\kb{d_2}{d_2}$ parametrized by a single parameter $p\in[0,1]$ whereas one needs two parameters for the set of states $\rho$ for which $\tr{\rho\sigma_z}=0$.
\end{remark}

\section{Conclusions}

In this paper we have identified some important optimality properties of a quantum observable represented mathematically as a POVM $\Mo$: Determination of the past (the pre-measurement state of the system), i.e.,\ informational completeness; freedom from dependence on more informational measurements from which the output data of $\Mo$ may be processed, i.e.,\ post-processing maximality; freedom from interference of different measurement schemes, i.e.,\ extremality; determination of values, i.e.,\ whether for any outcome set $X$ one can prepare the system in a state realizing a value from $X$ with arbitrarily high accuracy; determination of the future, i.e.,\ any measurement of $\Mo$ also works as a state preparator; and freedom from quantum noise, i.e.,\ pre-processing maximality. We have investigated these properties and generalized results known for discrete observables for more general observables. We have also found connections between these conditions: Pre-processing maximality and determination of future are equivalent and both are characterized by the rank-1 property of $\Mo$. Moreover, using a refinement procedure, one can replace an informationally complete (resp.\ extreme) $\Mo$ with an informationally complete (resp.\ extreme) rank-1 POVM $\Mo^{\bm1}$. Determination of values is equivalent with the norm-1 property (i.e.\ $\|\Mo(X)\|=1$ for all outcome sets $X$, $\Mo(X)\ne I_\hi$) and using our characterizations, we immediately see that, when $\Mo$ is preprocessing clean, it automatically defines its values.

We may conclude that there are two major lines of optimality for quantum observables: On one hand,
an observable may be informationally complete, and, at the same time, such an observable may also be free from all kinds of classical noise, i.e.,\ it may be extreme and post-processing clean simultaneously. On the other hand, an observable may define its values, i.e.,\ have the norm-1 property; especially, the observable subtype may be pre-processing clean. However, we are not aware of any norm-1 informationally complete observable; informational completeness requires properties that are strongly opposed to properties found in typical norm-1 observables: unsharpness, noncommutativity, and nonlocalizability, namely the inability to prepare systems into states yielding particular outcomes in the measurement with certainty. Thus the two main optimality criteria, informational completeness (determination of past) and determination of values, at its strongest in eigenvalue-1 observables, appear as complementary properties of a quantum observable. However, an observable may be free from classical noise (extreme rank-1) as well as from quantum noise (pre-processing maximality) simultaneously, in which case the observable is forced to be a rank-1 PVM, which automatically defines its values and the future of the quantum system (only measurements of such observables are preparative) but fails to determine the past of the system.

We have also discussed and reviewed results concerning joint and sequential measurements involving optimal observables. Especially, all the observables which are jointly measurable with a rank-1 observable $\Mo$ are smearings (post-processings) of $\Mo$, and for a jointly measurable pair $(\Mo,\Mo')$ of observables, where either one of the observables is extreme, there exists a unique joint observable giving $\Mo$ and $\Mo'$ as its margins.

\section*{Acknowledgements}

The authors would like to thank T.\ Heinosaari for his remarks on the manuscript and Y.\ Kuramochi for turning the authors' attention to Ref. \cite{Beukema2006}. E.\ H.\ also acknowledges financial support from the Japan Society for the Promotion of Science (JSPS) as an overseas postdoctoral fellow at JSPS.

\end{document}